%% file: main.tex
\documentclass[fleqn]{article}
\title{Ordered $k$-Median with Outliers and Fault-Tolerance}

\author{Shichuan Deng\thanks{IIIS, Tsinghua University, Email: \texttt{dsc15@mails.tsinghua.edu.cn}} \and Qianfan Zhang\thanks{IIIS, Tsinghua University, Email: \texttt{zqf18@mails.tsinghua.edu.cn}}}

\date{}

\usepackage{amsmath,mathrsfs}
\usepackage{amsfonts,bbm}
\usepackage{amssymb}
\usepackage{graphicx,url}
\usepackage[linesnumbered,ruled,vlined]{algorithm2e}
\usepackage[margin=1in]{geometry}
\usepackage{xspace}
\usepackage[dvipsnames]{xcolor}
\usepackage{amsthm}
\usepackage{enumitem}
\usepackage{hyperref}
\hypersetup{
    colorlinks=true,
    linkcolor=blue,
    urlcolor=cyan,
    citecolor=ForestGreen
}
\usepackage[capitalise]{cleveref}
\usepackage[notrig]{physics}

\DeclareMathOperator*{\argmin}{arg\,min}
\newcommand*{\E}{\mathbb{E}}

\newcommand{\Ball}[2]{\mathsf{Ball}_{#1}\left(#2\right)}

\newcommand{\opt}{\mathsf{OPT}}
\newcommand{\fopt}{F^\star}
\newcommand{\copt}{C^\star}
\newcommand{\cpt}{C'^\star}
\newcommand{\da}{\downarrow}

\newcommand{\calI}{\mathcal{I}}

\newcommand{\calC}{\mathcal{C}}
\newcommand{\calF}{\mathcal{F}}

\newcommand{\calB}{\mathcal{B}}
\newcommand{\calL}{\mathcal{L}}
\newcommand{\calU}{\mathcal{U}}

\newcommand{\calM}{\mathcal{M}}
\newcommand{\scrC}{\mathscr{C}}
\newcommand{\scrF}{\mathscr{F}}
\newcommand{\scrJ}{\mathscr{J}}
\newcommand{\scrI}{\mathscr{I}}

\newcommand{\ball}{\mathsf{Ball}}

\newcommand{\ww}{\widetilde{w}}
\newcommand{\hw}{\hat{w}}
\newcommand{\ow}{\overline{w}}
\newcommand{\wt}{\mathsf{wt}}

\newcommand{\nextp}[1]{\mathsf{next}(#1)}
\newcommand{\que}{\mathcal{Q}}

\newcommand{\cfull}{C_{\text{full}}}
\newcommand{\cpart}{C_{\text{part}}}
\newcommand{\ccore}{C_{\text{core}}}

\newcommand{\pos}{\mathsf{POS}}
\newcommand{\topl}[2]{\mathrm{Top}_{#1}(#2)}

\newcommand{\vcc}{\vec{c}}
\newcommand{\vco}{\vec{o}}
\newcommand{\vce}{\vec{\xi}}
\newcommand{\dav}{d_\mathsf{av}}
\newcommand{\davt}[1]{{\calL_{#1}}_\mathsf{av}}
\newcommand{\dmax}{d_\mathsf{max}}
\newcommand{\dmin}{d_\mathsf{min}}
\newcommand{\etal}{\textit{et~al.}\xspace}
\newcommand{\etalcite}[1]{\textit{et~al.}~\cite{#1}}

\newcommand{\orderedcluster}{{\small\textsc{OrdClst}}\xspace}
\newcommand{\ftkc}{{\small\textsc{FT$k$Cen}}\xspace}
\newcommand{\ftkm}{{\small\textsc{FT$k$Med}}\xspace}
\newcommand{\rkc}{{\small\textsc{R$k$Cen}}\xspace}
\newcommand{\rkm}{{\small\textsc{R$k$Med}}\xspace}
\newcommand{\om}{{\small\textsc{O$k$Med}}\xspace}

\newcommand{\oftm}{\mathsf{FTO\mathit{k}Med}}
\newcommand{\oko}{\mathsf{RO\mathit{k}Med}}
\newcommand{\omm}{\mathsf{OMatMed}}
\newcommand{\okm}{\mathsf{OKnapMed}}

\newtheorem{theorem}{Theorem}
\newtheorem{lemma}[theorem]{Lemma}

\newtheorem{fact}[theorem]{Fact}
\newtheorem{property}[theorem]{Property}
\newtheorem*{informal}{Informal Theorem}

\theoremstyle{definition}
\newtheorem{definition}[theorem]{Definition}
\theoremstyle{remark}
\newtheorem*{remark*}{Remark}
\newtheorem*{note*}{Note}

\begin{document}
\maketitle
\input{intro}

\input{robust}
\input{fault}

\bibliography{references}

\input{appendix}

\end{document}

%% file: intro.tex
\begin{abstract}
	In this paper, we study two natural generalizations of the \emph{ordered $k$-median} problem, named \emph{robust ordered $k$-median} and \emph{fault-tolerant ordered $k$-median}. In ordered $k$-median, given a finite metric space $(X,d)$, we seek to open a set of $k$ facilities $S\subseteq X$ which induces a \emph{service cost vector} $\vcc=\{d(j,S):j\in X\}$, and minimize the \emph{ordered objective} $w^\top\vcc^\da$. Here $d(j,S)=\min_{i\in S}d(j,i)$ is the minimum distance between $j$ and open facilities in $S$, $w\in\mathbb{R}^{|X|}$ is a predefined non-increasing non-negative vector, and $\vcc^\da$ is the non-increasingly sorted version of $\vcc$. The current best result for ordered $k$-median is a $(5+\epsilon)$-approximation [Chakrabarty and Swamy, STOC 2019].
	
	We first consider robust ordered $k$-median, also known as ordered $k$-median with outliers, where the input consists of an ordered $k$-median instance and an additional parameter $m\in\mathbb{Z}_+$. The goal is to open $k$ facilities $S$, select $m$ clients $T\subseteq X$ and assign the nearest open facility to each client $j\in T$. In this case, the service cost vector is $\vcc=\{d(j,S):j\in T\}$, the vector $w$ is in $\mathbb{R}^m$, and the objective is $w^\top\vcc^\da$. We introduce a novel yet simple objective function that enables linear analysis of the non-linear ordered objective, apply an iterative rounding framework [Krishnaswamy~\etal, STOC 2018] and obtain a constant-factor approximate solution. We devise the first constant-factor approximations for \emph{ordered matroid median} and \emph{ordered knapsack median} using the same method. 
	
	We also consider the fault-tolerant ordered $k$-median problem, where besides the same input as standard ordered $k$-median, we are also given additional client \emph{requirements} $\{r_j\in\mathbb{Z}_+:j\in X\}$ and need to assign $r_j$ distinct open facilities to each client $j\in X$. The service cost of $j$ is the sum of distances to its assigned facilities, and the objective is still $w^\top\vcc^\da$. We obtain a constant-factor approximation algorithm using a novel LP relaxation with constraints created via a new sparsification technique.
\end{abstract}

\section{Introduction}


$k$-center and $k$-median are two of the most fundamental clustering problems. In standard $k$-center, given a finite metric space $(X,d)$ and $k\in\mathbb{Z}_+$, the goal is to select at most $k$ open facilities $S\subseteq X$, and minimize the maximum distance from each client in $X$ to its nearest facility in $S$, i.e., $\max_{j\in X}d(j,S)$ where $d(j,S)=\min_{i\in S}d(j,i)$. $k$-center is NP-hard to approximate to a ratio smaller than 2~\cite{hsu1979bottleneck}, and classic 2-approximations are known~\cite{gonzalez1985clustering,hochbaum1985best}. In standard $k$-median, the input is the same as $k$-center and the objective is the sum of distances from each client to its nearest facility, i.e., $\sum_{j\in X}d(j,S)$. $k$-median is APX-hard~\cite{jain2002greedy}, while several constant-factor approximation algorithms are developed~\cite{arya2004local,charikar2002constant,charikar2012dependent,jain2001approximation,li2016approximating}. Currently the best ratio $(2.675+\epsilon)$ is due to Byrka~\etalcite{byrka2017improved}.

$k$-center and $k$-median represent the two extremes on the aggregation of client connection costs. In particular, let $\vcc=\{d(j,S):j\in X\}$ be the service cost vector. The objective of $k$-center is to minimize $\|\vcc\|_\infty$, and $k$-median requires to minimize $\|\vcc\|_1$. One naturally asks for a more general formulation that unifies the two. The ordered $k$-median problem (\om) is one such problem with ordered objective $w^\top\vcc^\da$, where $\vcc^\da$ is the non-increasingly sorted service cost vector and $w$ is a given non-increasing, non-negative vector. We also define $cost(w;\vec{v})=w^\top\vec{v}^\da$ for convenience.
\om unifies several fundamental clustering problems, including $k$-center, $k$-median, $k$-centdian and $k$-facility $p$-centrum. We refer the reader to the book~\cite{nickel2009location} by Nickel and Puerto for a more detailed overview. Several constant-factor approximations are developed in recent works~\cite{aouad2019ordered,byrka2018constant,chakrabarty2018interpolating}, and the current best ratio is $(5+\epsilon)$ due to Chakrabarty and Swamy~\cite{chakrabarty2019approximation}.

The ordered objective also arises in various other combinatorial problems. Fern\'{a}ndez~\etalcite{fernandez2017ordered} consider the ordered objective in multi-objective spanning tree optimization. Puerto~\etalcite{puerto2017revisiting} study a variety of graph problems with the ordered objective. Chakrabarty and Swamy~\cite{chakrabarty2019approximation,chakrabarty2019simpler} study load balancing with the more general objective of symmetric monotone norms\footnote{The ordered objective is a special case of symmetric monotone norms (see, e.g., \cite{chakrabarty2019approximation}).} and give constant factor approximations, which is recently improved to a $(2+\epsilon)$-approximation \cite{ibrahimpur2021minimum}.
Ibrahimpur and Swamy~\cite{ibrahimpur2020approximation} consider the objective of symmetric monotone norms on stochastic load balancing and spanning tree problems, and show several algorithms with varying approximation ratios. These are just a few examples, and we refer to the books~\cite{laporte2015location,nickel2009location} for a more comprehensive overview.

\subsection{Our Contributions}\label{section:contribution}


This paper is concerned with the exploration of the ordered objective in more general clustering settings. We first consider robust ordered $k$-median ($\oko$), where we are given a metric space $(X,d)$ and asked to open at most $k$ facilities $S$. We also need to select $m$ clients $T\subseteq X$, and assign the nearest open facility to each of them. The service cost vector is defined as $\vcc=\{d(j,S):j\in T\}$, and our goal is to minimize $cost(w;\vcc)=w^\top\vcc^\da$ for a given non-increasing non-negative vector $w\in\mathbb{R}^m$.
Robust clustering, also known as clustering with outliers, attracts a lot of research interest recently. In \emph{robust $k$-center} (\rkc) with the same input but no vector $w$, the objective is the maximum service cost of the $m$ chosen clients. Charikar~\etalcite{charikar2001algorithms} give a 3-approximation algorithm for \rkc. Chakrabarty~\etalcite{chakrabarty2020non} improve the result to a best-possible 2-approximation (also see Harris~\etalcite{harris2019lottery}). Chakrabarty and Negahbani~\cite{chakrabarty2018generalized} show a 3-approximation for robust center clustering, under certain constraints of a generalized, down-closed family of facility sets. In \emph{robust $k$-median} (\rkm), the objective is the sum of service costs of the $m$ chosen clients. Chen~\cite{chen2008constant} gives the first constant factor approximation for \rkm. Krishnaswamy~\etalcite{krishnaswamy2018constant} employ iterative rounding and obtain an approximation ratio of $(7.081+\epsilon)$, which is later improved to $(6.994+\epsilon)$ by Gupta~\etalcite{gupta2021structural}. $\oko$ generalizes \rkc and \rkm by choosing the vector $w$ accordingly.

We extend our results to ordered matroid median ($\omm$) and ordered knapsack median ($\okm$), which are natural generalizations of ordered $k$-median, by replacing the cardinality constraint $|S|\leq k$ with a matroid constraint and a knapsack constraint.\footnote{In the matroid version, a matroid $\calM=(X,\calI)$ is given with the constraint $S\in\calI$; in the knapsack version, each $i\in X$ has a weight $\wt_i\geq0$ and the constraint becomes $\wt(S)=\sum_{i\in S}\wt_i\leq W$ for a given budget $W\geq0$.}
They generalize \emph{matroid center} and \emph{matroid median}, \emph{knapsack center} and \emph{knapsack median}, respectively. 
Chen~\etalcite{chen2016matroid} show a 3-approximation for matroid center. Krishnaswamy~\etalcite{krishnaswamy2011matroid} give the first constant approximation for matroid median, and thereafter the ratio is improved in~\cite{charikar2012dependent,swamy2016improved} with the current best ratio 7.081 due to Krishnaswamy~\etalcite{krishnaswamy2018constant}.
Hochbaum and Shmoys study knapsack center and give a 3-approximation~\cite{hochbaum1986unified}. As for knapsack median, Kumar gives the first 2700-approximation in~\cite{kumar2012constant}. The ratio is later improved in~\cite{byrka2015knapsack,charikar2012dependent,krishnaswamy2018constant,swamy2016improved}, and the best ratio so far is $(6.387+\epsilon)$ due to Gupta~\etalcite{gupta2021structural}. To the best of our knowledge, no constant factor approximations are known for $\omm$ and $\okm$ prior to our study. 

Then, we consider the fault-tolerant ordered $k$-median problem ($\oftm$). Given the metric space $(X,d)$, we need to select $k$ open facilities $S$ and assign the nearest $r_j\geq1$ open facilities to each $j\in X$. The service cost vector is defined as $\vcc=\{d_{r_j}(j,S):j\in X\}$, in which $d_{r_j}(j,S)=\min_{S'\subseteq S,|S'|=r_j}\sum_{i\in S'}d(i,j)$ is the sum of distances from $j$ to its nearest $r_j$ facilities. Our goal is to minimize $cost(w;\vcc)=w^\top\vcc^\da$, where $w\in\mathbb{R}^{|X|}$ is a given non-increasing non-negative vector.
Fault-tolerant clustering problems have been studied extensively (see, e.g.,~\cite{hajiaghayi2016constant,inamdar2019fault,kumar2013fault}). In \emph{fault-tolerant $k$-center} (\ftkc) with the same input but no vector $w$, the service cost of $j$ is the distance between $j$ and its $r_j$-th nearest facility, and the objective is the maximum service cost. \ftkc generalizes standard $k$-center by setting all $r_j$s to 1. Chaudhuri~\etalcite{chaudhuri1998p} and Khuller~\etalcite{khuller2000fault} give tight 2-approximation algorithms for \ftkc.
In an identical setting, the \emph{fault-tolerant $k$-median} (\ftkm) problem defines service cost $\vcc_j$ as the sum of distances from $j$ to its $r_j$ nearest open facilities, with the objective being the sum of all service costs. \ftkm generalizes standard $k$-median by setting all $r_j$s to 1. Swamy and Shmoys~\cite{swamy2008fault} give a 4-approximation for the case of uniform $r_j$s, and Hajiaghayi~\etalcite{hajiaghayi2016constant} devise the first constant factor approximation algorithm for arbitrary $r_j$s, with an approximation ratio of 93.

We deliberately simplify the definitions in this section by referring to the metric space as $(X,d)$. Our actual results are more general, by considering $X=\calF\cup\calC$ and only allowing facilities to be located in $\calF$ and restricting to the clients in $\calC$. This is also known as the \emph{supplier version} of clustering problems. We summarize our results in the following informal theorem. 
\begin{informal}
There exist polynomial time constant factor approximation algorithms for $\oko$, $\omm$, $\okm$ and $\oftm$.
\end{informal}


\subsection{Our Techniques}\label{section:technique}

We focus on $\oko$ and $\oftm$ in this section. One main aspect of our technical contributions is the novel LP relaxations, where we improve upon previously-known sparsification techniques that are standard for tackling the ordered objective (e.g.,~\cite{aouad2019ordered,byrka2018constant,chakrabarty2018interpolating,chakrabarty2019approximation}). In what follows, we highlight several ideas in our LP design, and give a high-level overview of the rounding analyses.

\paragraph{Robust Ordered $k$-Median.}
We start with the sparsification method by Byrka~\etalcite{byrka2018constant}. In the pre-processing phase, one first guesses some disjoint intervals $I_0,I_1,\dots$, with each $I_t\subseteq\mathbb{R}_{\geq0}$ having the form $(x,(1+\epsilon)x]$ for some small $\epsilon>0$, so that the connection distances falling into the same interval differ by only a $(1+\epsilon)$ factor. Let $\overline{w}_I$ be the average weight multiplied with distances in the interval $I$ in a fixed optimal solution. If we apply the same weight $\overline{w}_I$ to all distances in $I$, the optimal solution exhibits a similar objective by only losing a $(1+O(\epsilon))$ factor. The pre-processing phase proceeds to build the premise that the guessed intervals $\{I_0,I_1,\dots\}$ and the guessed average weights $\{\overline{w}_{I_0},\overline{w}_{I_1},\dots\}$ approximately agree with the unknown optimum. This is done by showing the number of all possible ``approximate configurations'' is polynomial-bounded, thus a simple enumeration suffices. To ``pre-apply'' the weights in an LP relaxation, we define a function $f$ that takes $d(i,j)$ to $\overline{w}_I\cdot d(i,j),\,d(i,j)\in I$ and put $f(d(i,j))$ instead of $d(i,j)$ in the LP objective. Byrka~\etalcite{byrka2018constant} implicitly use such a function and give a $(38+\epsilon)$-approximation for \om.

Unlike in \om, this objective function appears to be incompatible with the iterative rounding framework we use for $\oko$ (adapted from~\cite{krishnaswamy2018constant}). Roughly speaking, one reason is that the pre-processing procedures of iterative rounding create a new instance and a corresponding new solution that has slightly larger service costs compared to the original problem. This can be easily handled in \rkm \cite{krishnaswamy2018constant} by losing a small factor since the contribution of each client is linear in its service cost, but the ordered objective in $\oko$ is non-linear.

We overcome this technical barrier by introducing a novel yet very simple objective function. We replace $f(d(i,j))$ with $f(\lambda d(i,j))$, where $\lambda\in(0,1)$ is a small constant parameter. Intuitively speaking, even if pre-processing scales up the service costs, their sizes are still manageable due to the pre-scaling of $\lambda$. We point out that for the optimal service cost vector $\vco$, the gap between each $f(\vco_j)$ and $f(\lambda\vco_j)$ may be $\omega(1/\lambda)$ since $f$ is piece-wise linear and non-decreasing. Nevertheless, we overcome this gap by obtaining a linear upper bound for any integral solution to the new relaxation. More specifically, let $\opt$ be the optimum of the original $\oko$ instance. We show that any integral solution with objective $V$, in the $\lambda$-scaled relaxation, induces a solution to the original problem with cost at most $\lambda^{-1}(V+O(1)\opt)$. Furthermore, we show that there exists an algorithm which produces an integral solution with objective $V=O(\lambda)\opt$. Combined with the previous observation, we obtain a final solution having cost $O(1/\lambda)$ times the optimum.

We demonstrate the versatility of this objective function by showing the first constant factor approximations for $\omm$ and $\okm$ via the aforementioned iterative rounding framework. Though Gupta~\etalcite{gupta2021structural} give a slightly improved iterative rounding method and obtain better approximation ratios for \rkm and knapsack median, this framework is not our main contribution and the improvement is likely to be small due to our deterministic metric discretization, thus we choose the original method in~\cite{krishnaswamy2018constant} for its simplicity of presentation. 

\paragraph{Fault-Tolerant Ordered $k$-Median.}
We start with another sparsification method in~\cite{chakrabarty2019approximation}, where one estimates the optimal service cost vector $\vco$ via guessing the entries at specific indices in \emph{sorted} $\vco^\da$, for instance $\{\vco^\da_1,\vco^\da_2,\vco^\da_4,\vco^\da_8,\dots\}$. For standard \om, it usually suffices (see, e.g.,~\cite{chakrabarty2019approximation}) to add more constraints to the LP relaxation or modify the objective function thereof. Chakrabarty and Swamy~\cite{chakrabarty2019approximation} give a $(5+\epsilon)$-approximation for \om using this method.

Since the service costs are sums of connection distances in $\oftm$, these methods are no longer valid. Instead, we take the counter-intuitive approach of guessing the \emph{individual connection distances} in the optimal solution.
Specifically, we define $\vce=\{d(i,j):i,j\text{ connected in the optimum}\}$, and guess the values of entries at specific indices in the \emph{sorted} vector $\vce^\da$, e.g., $\{\vce^\da_1,\vce^\da_2,\vce^\da_4,\vce^\da_8,\dots\}$. We note that similar techniques are also recently employed in load balancing \cite{chakrabarty2019simpler,ibrahimpur2021minimum}, where one guesses individual job sizes.
We then design a new LP relaxation that directly approximates the ordered objective, using a linear decomposition of the ordered objective (see~\cite{chakrabarty2019approximation} and \cref{eq:conic}), which itself stems from a majorization theorem by Goel and Meyerson~\cite{goel2006simultaneous}. We solve the relaxation using the ellipsoid algorithm, and stochastically round the solution via the (slightly modified) rounding algorithm by Hajiaghayi~\cite{hajiaghayi2016constant} for \ftkm. Our analysis of the approximation ratio requires a judicious examination of the LP solution, drawing ideas from the oblivious rounding analysis by Byrka~\etalcite{byrka2018constant}.


\subsection{Related Work}


The $k$-means problem also offers an intermediate-style objective between $k$-center and $k$-median by optimizing the $\|\cdot\|_2$ norm of the service cost vector. It is the most widely used clustering objective in practice when considered in Euclidean space, and the most popular algorithm for it is the Lloyd's algorithm~\cite{lloyd1982least}. For the Euclidean case, Kanungo~\etalcite{kanungo2004local} show a $(9+\epsilon)$-approximation based on local search. The ratio is later improved to $(6.357+\epsilon)$ by Ahmadian~\etalcite{ahmadian2017better}, where they also give the currently best $(9+\epsilon)$-approximation for $k$-means in general metrics. There has been a huge body of literature on $k$-means heuristics and its variants, for example, variants of the Lloyd's algorithm~\cite{arthur2007advantage}, smoothed analysis of the algorithms~\cite{arthur2011smoothed} and clusterability of different instances~\cite{awasthi2010stability,cohen-addad2017local}, etc. As a more relevant result, Krishnaswamy~\etalcite{krishnaswamy2018constant} study the \emph{robust $k$-means} problem and give a $(53.002+\epsilon)$-approximation.


The \emph{uncapacitated facility location} problem (UFL) is closely related to $k$-median, and there has been a long line of research for both (e.g.,~\cite{arya2004local,charikar2012dependent,chudak2003improved,jain2001approximation,shmoys1997approximation}, also the book of Williamson and Shmoys~\cite{williamson2011design}). The input for UFL is the same as $k$-median except that the cardinality constraint $k$ is replaced by an opening cost $f_i\geq0$ for each facility $i$, and the objective is the sum of all opening costs and service costs. The current best approximation ratio is 1.488 due to Li~\cite{li2011approximation}, and the lower bound is 1.463~\cite{guha1998greedy} assuming $\mathrm{P}\neq\mathrm{NP}$. Charikar~\etalcite{charikar2001algorithms} study the \emph{robust UFL} problem and devise a 3-approximation.
In \emph{fault-tolerant facility location} (FTFL), the input is the same as UFL, except that each client $j$ needs to be assigned $r_j\geq1$ open facilities, and its service cost is the sum of distances to these facilities. Several constant factor approximations are developed~\cite{guha2003constant,swamy2008fault}, and the current best ratio is 1.725 due to Byrka~\etalcite{byrka2010fault}. Hajiaghayi~\etalcite{hajiaghayi2016constant} consider FTFL with arbitrarily weighted service costs and provide a 3.16-approximation.

%% file: robust.tex
\section{Robust Ordered $k$-Median}\label{section:robust}

In this section, we maintain a generic problem called \orderedcluster, which can be later instantiated as different concrete problems. An instance $\scrI$ of \orderedcluster consists of a latent facility set $\calF$, a latent client set $\calC$, a finite metric $d$ on $\calF\cup\calC$, feasible open facility sets $\mathscr{F}\subseteq2^\calF$, feasible served clients sets $\mathscr{C}\subseteq2^\calC$, and a given non-increasing non-negative vector $w\in\mathbb{R}^m$. Each $C\in\scrC$ satisfies $|C|=m$, and $d(u,v)\geq 1$ for $u,v \in \calF\cup\calC$ that are \emph{not} co-located.
The goal is to choose $F\in\scrF,\,C\in\scrC$ that induce a service cost vector $\vcc=(d(j,F):j\in C)$, minimizing the ordered objective $cost(w;\vcc)=w^\top\vcc^\da$. 


\subsection{Reduction from Ordered Metric to Non-Ordered Non-Metric}

We devise a general framework that reduces \orderedcluster to clustering problems that have linear objectives but at the expense of non-metric connection costs. Given an instance $\scrI=(\calF,\calC,d,\scrF,\scrC,w)$ of \orderedcluster and $f:\mathbb{R}_{\geq0}\rightarrow\mathbb{R}_{\geq0}$ a \emph{non-decreasing} function, we say $\scrJ=(\calF,\calC,d,\scrF,\scrC,f)$ is a \emph{reduced instance} of $\scrI$, and the goal of optimization is to choose $F\in\scrF$, $C\in\scrC$ such that the sum of assignment costs $\sum_{j\in C}f(d(j,F))$ is minimized.

The framework is similar to previous sparsification methods \cite{aouad2019ordered,byrka2018constant} for ordered objectives, but has additional properties necessary to overcome certain technical difficulties when instantiated with concrete formulations. For convenience, for each function $f:\mathbb{R}_{\geq0}\rightarrow\mathbb{R}_{\geq0}$ and $\lambda>0$, we define $f_\lambda(x)=f(\lambda x)$, $\forall x\geq0$.
We let $n_0=|\calF\cup\calC|$, fix a small constant $\epsilon>0$ and present the following core lemma. 
\begin{lemma}
  \label{theorem:ordered1}
  Let $\scrI=(\calF,\calC,d,\scrF,\scrC,w)$ be an instance of \orderedcluster. For any small $\epsilon>0$, there exists an algorithm that outputs $(n_0/\epsilon)^{O(1/\epsilon)}$ many non-decreasing functions $\mathbb{R}_{\geq0}\rightarrow\mathbb{R}_{\geq0}$, such that there exists an output $f$, satisfying that for each $\lambda\in(0,1]$, the reduced instance $\scrJ_{f,\lambda}=(\calF,\calC,d,\scrF,\scrC, f_\lambda)$ has an optimal objective of at most $\lambda(1+8\epsilon)\opt$.
  Further, if an algorithm produces a solution with objective $V$ for $\scrJ_{f,\lambda}$, the same solution attains an objective of at most $\lambda^{-1}(V+(1+4\epsilon)\opt)$ for $\scrI$. We say such $f$ is \emph{faithful}.
\end{lemma}

As a direct consequence of \cref{theorem:ordered1}, if one is to obtain a solution to any faithfully reduced instance $\scrJ_{f,\lambda}$ with objective at most $V\leq\gamma\opt$, the result is automatically a $\lambda^{-1}(\gamma+1+4\epsilon)$-approximation to $\scrI$.
Before we proceed with the proof, we discuss the sparsification method \cite{aouad2019ordered,byrka2018constant} used for constructing such functions.
Let $\vco\in\mathbb{R}^m$ be the service cost vector in a fixed (unknown) optimal solution $(\fopt,\copt)$ to the original instance $\scrI$, and $\opt=cost(w;\vco)$ be the optimal objective thereof. We first guess the exact value of $\vco^\da_1$, i.e., the largest service cost, which only has a polynomial number of possible values. Let $T$ be the smallest integer s.t. $\epsilon(1+\epsilon)^T > m$ and define intervals $I_{T+1},I_T,...,I_0$ where $I_{T+1}=[0,\frac{\epsilon \vco^\da_1}{m}]$, $I_t=(\frac{\epsilon \vco^\da_1}{m}(1+\epsilon)^{T-t},\frac{\epsilon \vco^\da_1}{m}(1+\epsilon)^{T-t+1}]$ for $t\in[T]$ and $I_0=(\frac{\epsilon \vco^\da_1}{m}(1+\epsilon)^T, \infty)$. Since $\bigcup_{t=0}^{T+1}I_t = [0,+\infty)$ and they are mutually disjoint, each $d(i,j)$ falls into exactly one interval. To avoid technical difficulties caused by weights that are too small, we define a new vector $\ww$ where $\ww_i=\max\{w_i, \frac{\epsilon w_1}{m}\},\,i\in[m]$. We obtain the following simple fact, by observing $\ww\geq w$ and $cost(\ww;\vec{v})-cost(w;\vec{v})\leq m\cdot \frac{\epsilon w_1}{m}\cdot \vec{v}_1^\da\leq \epsilon\cdot cost(w;\vec{v})$.
\begin{fact}
  \label{lemma:pseudo1}
  For $\vec{v}\subseteq \mathbb{R}_{\geq0}^m$, one has $cost(w;\vec{v}) \le cost(\ww;\vec{v}) \le (1+\epsilon)cost(w;\vec{v})$.
\end{fact}

We consider the entries of $\vco$ that fall into intervals $I_{T+1},I_T,...,I_0$, and define the average weights $\ow_t$ w.r.t. $\vco$ and interval $I_t$, such that $\ow_0=\ww_1$ and
\begin{equation}
  \ow_t = \begin{cases}
    \left.\left(\sum_{i:\vco^\da_i\in I_t} \ww_i\right)\right/|\vco \cap I_t| & \vco \cap I_t \ne \emptyset,\, t \ge 1,\\
    \ow_{t-1} & \vco \cap I_t = \emptyset,\, t \ge 1.
  \end{cases}\label{eqn:average:weights}
\end{equation}
Since $\ww$ is non-increasing, it follows that $\ow$ is also non-increasing. Though the actual sequence $\{\ow_t>0\}_{t=0}^{T+1}$ is unknown, we may guess by using another non-increasing sequence $\{\hw_t>0\}_{t=0}^{T+1}$
such that for each $0 \le t \le T+1$, $\hw_t$ is an integer power of $(1+\epsilon)$. Since the entries of $\ow$ are at least $\min_i\ww_i\geq\epsilon w_1/m$ and at most $w_1$, the number of possible values for $\hw$ is $O(\log_{1+\epsilon}(m/\epsilon))$. By definition of $T$, we also have $T=O(\log_{1+\epsilon}(m/\epsilon))$, so using routine calculations, the number of all possible non-increasing sequences is at most $(m/\epsilon)^{O(1/\epsilon)}$.
Up to now, we have only guessed the value of $\vco^\da_1$ and the average weights $\hw_t$, hence the total number of possible guesses is at most $(n_0/\epsilon)^{O(1/\epsilon)}$ since $m\leq n_0$.

\begin{proof}[Proof of \cref{theorem:ordered1}]
For each guess $(\vco^\da_1,\{\hw_t\}_{t=0}^{T+1})$, we define a piece-wise linear function $f$ where
\begin{equation}
  f(x) = \hw_t x,\; x\in I_t,\; 0 \le t \le T+1.\label{eq:costfunction}
\end{equation}
Because $\{\hw_t\}_t$ is non-increasing, $f$ is non-decreasing and non-negative.
To prove the lemma, one need only show the existence of a faithful function. In the sequel, we assume that the guessed premise is as desired, that is, $\vco^\da_1$ is precisely the largest connection distance in the optimal solution and for each $0\leq t\leq T+1$, one has $\hw_t\in[\ow_t, (1+\epsilon)\ow_t)$. We show that the corresponding function $f$ is faithful.
We first need the following two lemmas.
\begin{lemma}\label{lemma:pseudo3}
  For any $\lambda\in(0,1]$ and any $F\in\scrF$, $C\in\scrC$, let $\vcc=(d(j,F):j\in C)\in\mathbb{R}_{\geq0}^m$. One has $\lambda\cdot\ww^\top\vcc^\da \leq \sum_{j\in C} f(\lambda \vcc_j) + ((1+\epsilon)^2+\epsilon)\opt$.
\end{lemma}
\begin{proof}
Recall that $|C|=m$ for each feasible client set, thus we assume $C=[m]$ for convenience.
Consider each $j\in[m]$ s.t. $\lambda\ww_j\vcc_j^\da > f(\lambda \vcc_j^\da)$. Notice that $\lambda \vcc_j^\da \notin I_0$, otherwise one has $\hw_0\ge \ow_0=\ww_1\ge\ww_j$, and thus $\lambda\ww_j\vcc^\da_j>f(\lambda\vcc^\da_j)=\hw_0(\lambda\vcc^\da_j)\geq\lambda\ww_j\vcc^\da_j$, which is a contradiction. If $\lambda \vcc_j^\da \in I_{T+1}=[0,\epsilon \vco^\da_1/m]$, one has $\lambda\ww_j\vcc_j^\da \le \ww_j(\epsilon \vco^\da_1/m)\leq \epsilon\cdot w^\top\vco^\da/m$.

Then, suppose $\lambda \vcc_j^\da \in I_t,\,t\in [T]$. We claim $\lambda \vcc_j^\da \le (1+\epsilon)\vco^\da_j$. For the sake of contradiction, assume otherwise and $\lambda \vcc_j^\da > (1+\epsilon)\vco^\da_j$, thus $\lambda\vcc^\da_j$ and $\vco^\da_j$ must be in different intervals. Suppose $\vco^\da_j \in I_{t'}$ for some $t' > t$, which implies $\ow_t \ge \ww_j$, because $\ow_t$ is the average weight on $I_t$, and $\ww_j$ is the weight for $\vco^\da_j\in I_{t'}$. As a result, we have $f(\lambda \vcc_j^\da) = \hw_t(\lambda \vcc_j^\da) \ge \lambda\ow_t \vcc_j^\da \ge \lambda\ww_j \vcc_j^\da$, contradicting our initial assumption. Thus the claim is true. 
  
The above analysis shows that $\lambda\ww_j\vcc^\da_j\leq f(\lambda\vcc^\da_j)+(1+\epsilon)\ww_j\vco^\da_j+\epsilon w^\top\vco^\da/m$ for each $j\in[m]$.
We sum over $j\in[m]$ and obtain
  \begin{equation*}
    \lambda\cdot\ww^\top\vcc^\da = \lambda\sum_{j\in[m]}\ww_j\vcc^\da_j \le \sum_{j\in [m]} f(\lambda \vcc_j) + (1+\epsilon) cost(\ww;\vco) + \epsilon\cdot cost(w;\vco).
  \end{equation*}
  Combining with Fact~\ref{lemma:pseudo1} and $\opt=cost(w;\vco)$, the lemma follows.
\end{proof}

\begin{lemma}\label{lemma:pseudo2}
    For any $\lambda\in(0,1]$, $\sum_{j\in C^\star} f(\lambda \vco_j) \le \lambda((1+\epsilon)^3+\epsilon) \opt$.
\end{lemma}
\begin{proof}
Recall that $(\fopt,\copt)$ is optimal for $\scrI$.
Consider any non-empty $\vco\cap I_t$, and it is easy to verify that $t\neq0$. Since $\lambda \leq 1$, some entries may be shifted to $I_{t'}$ with $t' > t$. If $t\leq T$, the contribution of $\lambda(\vco\cap I_t)$ on the LHS is at most
\begin{equation*}
    \sum_{\substack{j:\vco^\da_j\in I_t,\\\lambda \vco^\da_j\in I_{>t}}} \lambda \hw_{>t}\vco^\da_j
    +\sum_{\substack{j:\vco^\da_j\in I_t,\\\lambda \vco^\da_j\in I_{t}}} \lambda \hw_t \vco^\da_j
    \leq\lambda\sum_{j:\vco^\da_j\in I_t}(1+\epsilon)\ow_t\vco^\da_j\leq\lambda(1+\epsilon)^2\sum_{j:\vco^\da_j\in I_t}\ww_j\vco^\da_j,
\end{equation*}
where the first inequality is due to non-increasing $\hw$, $\hw_t\in[\ow_t,(1+\epsilon)\ow_t)$ and the second inequality is because within the same interval $I_1,\dots,I_T$, the values of $\vco^\da_j$ differ by a factor no more than $(1+\epsilon)$. If $t=T+1$, each such $\vco^\da_j\leq\epsilon\vco^\da_1/m$, thus the contribution of $\lambda(\vco\cap I_{T+1})$ is at most $\epsilon\lambda\ww_1\vco^\da_1$.
By summing over all non-empty $\vco\cap I_t$ and using Fact~\ref{lemma:pseudo1}, the lemma follows.
\end{proof}
Now back to the original lemma and fix any $\lambda\in(0,1]$. For the first assertion, since $(\fopt,\copt)$ is a feasible solution to both $\scrI$ and $\scrJ_{f,\lambda}$, the conclusion follows from \cref{lemma:pseudo2}. 
For the second assertion, let $(F,C)$ be the solution returned by the algorithm, thus $V=\sum_{j\in C}f(\lambda d(j,F))$. Therefore, using Fact~\ref{lemma:pseudo1} and \cref{lemma:pseudo3}, the objective of this solution in the \orderedcluster instance $\scrI$ is at most $cost(w;\vcc)\leq cost(\ww;\vcc)\leq\lambda^{-1}(V+(1+4\epsilon)\opt)$.
\end{proof}


\subsection{Application: Robust Ordered $k$-Median}

In $\oko$, \orderedcluster is formulated such that $\scrF$ is the rank-$k$ uniform matroid on $\calF$ and $\scrC=\{C\subseteq\calC:|C|=m\}$. Via \cref{theorem:ordered1} and enumeration of all possible functions, suppose that we have a faithful function $f$ in what follows. To obtain a good solution to some $\scrJ_{f,\lambda}$, we adapt the iterative rounding framework for robust $k$-median by Krishnaswamy~\etalcite{krishnaswamy2018constant}.
Denote $x_{ij}\in[0,1]$ the extent of connecting client $j$ to facility $i$, and $y_i\in[0,1]$ the extent of opening facility $i$. The natural relaxation for $\scrJ_{f,\lambda}$ is defined as follows.
\begin{alignat*}{3}
  \text{min\quad}&&\sum_{j \in \calC}\sum_{i\in \calF} x_{ij} f(\lambda d(i,j)) \tag{$\mathrm{LP}(f,\lambda)$}\label{lp:1}\\
  \text{s.t.\quad}&&\sum_{j\in \calC}\sum_{i \in \calF} x_{ij} &\ge m\\
  && \sum_{i \in \calF} x_{ij} &\le 1 && \qquad\forall j \in \calC\\  
  &&\sum_{i\in\calF}y_i&=k\\
  &&0\leq x_{ij} \leq y_i &\leq 1  && \qquad\forall i \in \calF, j \in \calC.
\end{alignat*}

\paragraph{Pre-processing.}
Instead of directly solving \ref{lp:1}, we employ pre-processing techniques and simplify the instance. The first idea is to include a constant number of facilities $S_0$ as must-have choices and remove some clients in advance, so that the new instance $\scrJ'_{f,\lambda}$ on remaining clients $\calC'$ (and properly modified $\scrC'$) is easier to solve, and a straightforward greedy algorithm on the removed clients $\calC\setminus\calC'$ maintains a good approximate solution to $\scrJ_{f,\lambda}$.

The second ingredient is a strengthened relaxation that has the same objective as \ref{lp:1}, except for a slightly larger coefficient $\lambda'>\lambda$. It has both the constraints of \ref{lp:1}, and additional constraints that guarantee certain sparse features in its solutions. We consider the new relaxation as associated with the instance $\scrJ'_{f,\lambda'}$. We skip \ref{lp:1} and solve the stronger relaxation instead.


\paragraph{Iterative rounding.}
After obtaining a fractional optimal solution to the new relaxation (hence $\scrJ'_{f,\lambda'}$ and $\scrJ_{f,\lambda'}$), we use the iterative rounding algorithm and obtain an integral solution, which translates naturally to a solution to $\scrI$. However, because the function $f$ in the LP objective is non-linear, we cannot directly estimate the objective of the solution through the new relaxation and $\lambda'$. Nevertheless, the integral solution is still feasible to \ref{lp:1} where the coefficient $\lambda$ is smaller, making it possible for us to bound the objective under the instance $\scrJ_{f,\lambda}$. We invoke \cref{theorem:ordered1} on $\scrJ_{f,\lambda}$ for the final approximation ratio.

\paragraph{Result.} The following theorem is one of the main results of this paper. We provide details of the algorithm and the proof in \cref{appendix:robust}.

\begin{theorem}\label{theorem:robust-main}
There exists a polynomial time $127$-approximation for $\oko$.
\end{theorem}


\subsection{Application: Ordered Matroid/Knapsack Median}

In $\omm$, we instantiate \orderedcluster with $\scrF$ being the set of independent sets of an arbitrary matroid on $\calF$ and $\scrC=\{\calC\}$. In $\okm$, each facility in $\calF$ has a weight $\wt_i\geq0$, $\scrF$ is the set of facility subsets with total weight at most $W$, and $\scrC=\{\calC\}$. The following theorem arises from the same reduction by \cref{theorem:ordered1} and a similar iterative rounding algorithm as $\oko$. We provide proofs and details of the algorithms in \cref{appendix:matroid} and \cref{appendix:knapsack}.

\begin{theorem}
There exist a polynomial time $19.8$-approximation for $\omm$, and a polynomial time $41.6$-approximation for $\okm$.
\end{theorem}

%% file: fault.tex
\section{Fault-Tolerant Ordered $k$-Median}\label{section:fault}

In this section, we consider the $\oftm$ problem. We first consider a simpler ``$\mathrm{Top}\text{-}\ell$ version'' of the problem, where the given non-increasing non-negative vector $w$ consists of only 1s and 0s, and thus the objective sums up the $\ell$ largest service costs. For the $\mathrm{Top}\text{-}\ell$ version, we obtain a randomized constant-factor approximation using LP rounding. Then for $\oftm$ with general input vectors, using the following observation \cite{byrka2018constant,chakrabarty2019approximation,goel2006simultaneous} where $n=|\calC|$, one decomposes the ordered objective $cost(w;\vcc)$ into a conic combination of $\mathrm{Top}\text{-}\ell$ objectives.
\begin{equation}
	\sum_{\ell\in[n]}(w_\ell-w_{\ell+1})\topl{\ell}{\vcc}
	=\sum_{\ell\in[n]}(w_\ell-w_{\ell+1})\sum_{\ell'\leq\ell}\vcc^\da_{\ell'}
	=\sum_{\ell\in[n]}w_\ell\vcc^\da_\ell
	=cost(w;\vcc).\label{eq:conic}
\end{equation}
Using a new LP relaxation and the linearity of expectation, the same LP rounding algorithm provides a similar constant-factor approximation guarantee.

\subsection{The Simple Top-$\ell$ Case}

We start with a simpler problem, where the given vector $w$ has $\ell$ leading 1s followed by 0s, $\ell\in[n]$. Equivalently, the objective is the so-called ``top-$\ell$ norm'' $\mathrm{Top}_\ell(\vcc)$, by adding up the $\ell$ largest entries in the service cost vector $\vcc$. Fix an unknown optimal solution with optimum $\opt$, and let the vector of \emph{individual connection costs} be $\vce=\{d(i,j):i,j\text{ connected in the optimum}\}$. Since there are only a polynomial number of possible values for $\vce^\da_\ell$, that is, the $\ell$-th largest entry in $\vce$, we enumerate all possible values and assume we have guessed the correct one $T_\ell=\vce^\da_\ell$ in what follows.

We first define the relaxation for this Top-$\ell$ version of $\oftm$. Since in the optimal solution, $\opt$ is the sum of the $\ell$ largest service costs, and $T_\ell$ is the $\ell$-th largest \emph{individual connection cost}, it is easy to see that the sum of individual connection costs that are larger than $T_\ell$ is at most $\opt$. This simple observation motivates a new constraint in the relaxation (see \cref{lp:t1} below).

Let $x_{ij}\in[0,1]$ denote the extent of connection between facility $i\in\calF$ and client $j\in\calC$, and $y_i\in[0,1]$ denote the extent we open $i\in\calF$. For any $T\geq0$, define $\calL_{T}(i,j)=d(i,j)$ if $d(i,j)\geq T$ and 0 otherwise. The following relaxation contains an exponential number of constraints in \cref{lp:t2}, yet given any $(x,y,R_\ell)$, an efficient separation oracle need only check whether the $\ell$ clients with the largest costs $\sum_{i\in\calF}x_{ij}d(i,j)$ make for a sum larger than $R_\ell$. Thus, \ref{lp:t0} can be efficiently solved using the ellipsoid algorithm.
\begin{alignat*}{3}
	\text{min\quad}&&R_\ell&\geq0 \tag{$\mathrm{Top\text{-}LP}$}\label{lp:t0}\\
	\text{s.t.\quad}&&\sum_{j\in\calC}\sum_{i\in\calF}x_{ij}\calL_{1.001\cdot T_\ell}(i,j)&\leq R_\ell\tag{$\mathrm{Top\text{-}LP}$.1}\label{lp:t1}\\
	&&\sum_{j\in S}\sum_{i\in\calF}x_{ij}d(i,j)&\leq R_\ell &&\qquad\forall S\subseteq\calC,|S|=\ell\tag{$\mathrm{Top\text{-}LP}$.2}\label{lp:t2}\\
	&&\sum_{i\in\calF}x_{ij}&=r_j &&\qquad\forall j\in\calC\tag{$\mathrm{Top\text{-}LP}$.3}\label{lp:t3}\\
	&&\sum_{i\in\calF}y_i&=k\tag{$\mathrm{Top\text{-}LP}$.4}\label{lp:t4}\\
	&&0\leq x_{ij}\leq y_i&\leq1 &&\qquad\forall i\in\calF,j\in\calC.\tag{$\mathrm{Top\text{-}LP}$.5}\label{lp:t5}
\end{alignat*}

\begin{lemma}\label{lemma:feasible:ft:top}
    \ref{lp:t0} is feasible and the optimal solution satisfies $R_\ell\leq\opt$.
\end{lemma}
\begin{proof}
Let $(x^\star,y^\star,\opt)$ be the solution where, in the optimal solution, $x_{ij}^\star=1$ if client $j$ is connected to facility $i$ and 0 otherwise, $y_{i}^\star=1$ if facility $i$ is open and 0 otherwise. It satisfies the last three sets of constraints in \ref{lp:t0} by definition.

If $T_\ell=0$, \cref{lp:t1} becomes $\sum_{i,j}x_{ij}d(i,j)\leq R_\ell$ and \cref{lp:t2} is obsolete. Because we have $\vce^\da_\ell=T_\ell=0$, it is easy to verify that $\opt$ is equal to the sum of \emph{all} individual connection costs, hence \cref{lp:t1} is satisfied, and the lemma follows. We assume $T_\ell>0$ in what follows.

Since $1.001\cdot T_\ell>T_\ell$, the LHS of \cref{lp:t1} is the sum of all individual connection costs that are at least $1.001\cdot T_\ell$, which is at most $\opt$ according to the above analysis. In \cref{lp:t2}, the maximum LHS is exactly the sum of the $\ell$ largest service costs in the optimum, thus $(x^\star,y^\star,\opt)$ satisfies these constraints as well, whence the lemma follows.
\qed
\end{proof}

We obtain an optimal solution $(x,y,R_\ell)$ to~\ref{lp:t0}, and denote $\calF_j=\{i\in\calF:x_{ij}>0\}$ the (fractionally) connected facilities for $j$. Using standard facility duplication, e.g., \cite{charikar2012dependent,hajiaghayi2016constant}, we assume $x_{ij}\in\{0,y_i\}$ for each $i\in\calF,j\in\calC$, and $\calF_j$ contains the nearest $r_j$ volume of facilities to $j$, that is, $\sum_{i\in\calF_j}y_i=r_j$. Define $\dav(j)=r_j^{-1}\sum_{i\in\calF_j}x_{ij}d(i,j)$, thus $r_j\dav(j)$ represents the service cost of $j$ in the LP solution. We follow the algorithm in~\cite{hajiaghayi2016constant}, but construct a slightly different auxiliary LP for stochastic rounding. We partition the facilities into disjoint subsets of $\calF$ called ``bundles'', such that each has unit volume of facilities, and each client $j$ keeps a queue of $r_j$ distinct bundles. The stochastic rounding opens one facility in each bundle, such that the sum of distances from $j$ to the open facilities in its queue does not deviate much from its service cost $r_j\dav(j)$.

Suppose the stochastically rounded integral solution is $z'\in\{0,1\}^{\calF}$. Let $\hat F=\{i\in\calF:z_i'=1\}$ be the open facilities and $\vcc=\{d_{r_j}(j,\hat F):j\in\calC\}$ be the service cost vector. Recall that $\calL_{T}(i,j)$ evaluates to $d(i,j)$ if $d(i,j)\geq T$ and 0 otherwise. We have the following core lemma on each $d_{r_j}(j,\hat F)$. 

\begin{lemma}\label{lem-core}
	For any $T\geq0$, there exists a random variable $X_j$ s.t.
	$d_{r_j}(j,\hat F)\leq3r_j\dav(j)+435.6T+X_j$
	holds with prob. 1 and $\E[X_j]\leq226.7\sum_{i\in\calF}x_{ij}\calL_{T}(i,j)$.
\end{lemma}

We provide the stochastic rounding algorithm and the rather extended proof of \cref{lem-core} in \cref{appendix:fault}. Next, we prove the main result of this section using \cref{lem-core}.

\begin{theorem}\label{theorem-top-l}
    There exists a randomized polynomial time $666$-approximation for $\mathrm{Top}\text{-}\ell$ version of $\oftm$.
\end{theorem}
\begin{proof}
We show that the stochastic solution $\hat F$ satisfies the approximation ratio in expectation.
Let $T=1.001\cdot T_\ell$ in \cref{lem-core} and recall $T_\ell=\vce^\da_\ell$, i.e., the $\ell$-th largest individual connection cost in the optimal solution. Since $\topl{\ell}{\vcc}$ only contains the service costs of $\ell$ clients, using \cref{lem-core}, one has
\begin{align*}
    &\E[\topl{\ell}{\vcc}]\leq\max_{\substack{S\subseteq\calC\\|S|=\ell}}\sum_{j\in S}(3r_j\dav(j)+435.6\cdot 1.001\cdot T_\ell)+\sum_{j\in\calC}\E[X_j]\\
    &\leq3\max_{\substack{S\subseteq\calC\\|S|=\ell}}\sum_{j\in S}\sum_{i\in\calF}x_{ij}d(i,j)+436.1\ell\cdot T_\ell+226.7\sum_{j\in\calC}\sum_{i\in\calF}x_{ij}\calL_{1.001\cdot T_\ell}(i,j)\\
    &\leq229.7R_\ell+436.1\ell\cdot T_\ell,
\end{align*}
where we use the feasibility of $(x,y,R_\ell)$ and the constraints \cref{lp:t1} and \cref{lp:t2}. Further, the sum of the $\ell$ largest service costs in the optimal solution is at least the sum of the $\ell$ largest individual connection costs in $\vce$, thus $\opt\geq\vce^\da_1+\cdots+\vce^\da_\ell\geq\ell\cdot \vce^\da_\ell=\ell\cdot T_\ell$. Combined with the inequality above and \cref{lemma:feasible:ft:top}, we obtain $\E[\topl{\ell}{\vcc}]\leq666\cdot\opt$, whence the theorem follows.
\end{proof}


\subsection{The General Case}

For $\oftm$ with general input vectors, we consider a relaxation similar to \ref{lp:t0}, but with an objective that simulates the ordered objective using the decomposition in \cref{eq:conic}. 

In the pre-processing step, we use the sparsification methods of~\cite{chakrabarty2019approximation}. Let $n=|\calC|,\delta>0$ and define $\pos_{n,\delta}=\{\min\{\lceil(1+\delta)^s\rceil,n\}:s\geq0\}\subseteq[n]$. For $\ell\in\pos_{n,\delta}$, let $\nextp{\ell}$ be the smallest element in $\pos_{n,\delta}$ that is larger than $\ell$ with $\nextp{n}=n+1,\,\nextp{0}=1$. We abbreviate $\pos_{n,\delta}$ to $\pos$ when $n$ and $\delta$ are clear from the context. Given $w$ and $\pos$, define $\ww_i=w_i$ for $i\in\pos$. For $i\in[n]\setminus\pos$ and $\ell\in\pos$ s.t. $\ell<i<\nextp{\ell}$, define $\ww_i=w_{\nextp{\ell}}$, and $w_{n+1}=\ww_{n+1}=0$. We have the following simple lemma due to Chakrabarty and Swamy~\cite{chakrabarty2019approximation}.

\begin{lemma}\label{lem-sparsify}(\cite{chakrabarty2019approximation})
	For $\vec{v}\in\mathbb{R}_{\geq0}^n$, $cost(\ww; \vec{v})\leq cost(w;\vec{v})\leq (1+\delta)cost(\ww;\vec{v})$.
\end{lemma}

We consider a fixed optimum and individual distances $\vce=\{d(i,j):i,j\text{ connected in the optimum}\}$.
We try to guess the entries $\vce^\da_\ell,\,\ell\in\pos$. We first guess $\vce^\da_1$, which has a polynomial number of possible values, and define $T_\epsilon=\{\vce^\da_1(1+\epsilon)^{-s}:s\geq0\}\cap[\epsilon \vce^\da_1/n,+\infty)$ for a small $\epsilon>0$. We then guess a non-increasing sequence $\{T_\ell'\}_{\ell\in\pos}$ with entries from $T_\epsilon\cup\{0\}$, so that
$T_\ell'\in[\vce^\da_\ell,(1+\epsilon)\vce^\da_\ell)$ if $\vce^\da_\ell\geq\epsilon \vce^\da_1/n$ and 0 otherwise. Since $|T_\epsilon\cup\{0\}|=O(\log_{1+\epsilon}(n/\epsilon))$ and $\{T_\ell'\}_{\ell\in\pos}$ is non-increasing with length $|\pos|=O(\log_{1+\delta}n)$, using a basic counting method, the number of all possible non-increasing sequences $\{T_\ell'\}_{\ell\in\pos}$ is at most
$2^{O(\log_{1+\delta}n+\log_{1+\epsilon}(n/\epsilon))}=(n/\epsilon)^{O(1/\epsilon+1/\delta)}$.

Assume we have guessed the correct $\vce^\da_1$ and sequence $\{T_\ell'\}_{\ell\in\pos}$ with the desired properties. Define $T_\ell=T_\ell'+\epsilon \vce^\da_1/n$ for each $\ell\in\pos$, and thus $\vce^\da_\ell< T_\ell < (1+\epsilon)\vce^\da_\ell+\epsilon \vce^\da_1/n$ holds for all $\ell\in\pos$. 
With the guessed values $\{T_\ell\}_{\ell\in\pos}$ and sparsified weights $\ww$, our LP relaxation is given as follows, where $y_i$ is the extent we open facility location $i$ and $x_{ij}$ is the extent of connecting client $j$ to facility $i$. With similar constraints as \ref{lp:t0}, we can solve the relaxation using the ellipsoid algorithm.
\begin{alignat*}{3}
	\text{min\quad}&&\sum_{\ell\in\pos}(\ww_\ell-\ww_{\nextp{\ell}})R_\ell \tag{$\mathrm{FT}(\ww)$}\label{lp:00}\\
	\text{s.t.\quad}&&\sum_{j\in\calC}\sum_{i\in\calF}x_{ij}\calL_{T_\ell}(i,j)&\leq R_\ell &&\qquad\forall \ell\in\pos\tag{$\mathrm{FT}(\ww)$.1}\label{lp:01}\\
	&&\sum_{j\in S}\sum_{i\in\calF}x_{ij}d(i,j)&\leq R_\ell &&\qquad\forall S\subseteq\calC,|S|=\ell,\ell\in\pos\tag{$\mathrm{FT}(\ww)$.2}\label{lp:02}\\
	&&\sum_{i\in\calF}x_{ij}&=r_j &&\qquad\forall j\in\calC\tag{$\mathrm{FT}(\ww)$.3}\label{lp:03}\\
	&&\sum_{i\in\calF}y_i&=k\tag{$\mathrm{FT}(\ww)$.4}\label{lp:04}\\
	&&0\leq x_{ij}\leq y_i&\leq1 &&\qquad\forall i\in\calF,j\in\calC.\tag{$\mathrm{FT}(\ww)$.5}\label{lp:05}
\end{alignat*}

\begin{lemma}\label{lemma:feasible:ft}
	Denote $\opt$ the optimum of the original problem. \ref{lp:00} is feasible and has an optimal value at most $\opt$.
\end{lemma}
\begin{proof}
	Let $O\subseteq\calF$ be the optimal solution with $\vco=\{d_{r_j}(j,O):j\in\calC\}$ being the service cost vector. Define $x_{ij}^\star=1$ if $i$ is connected to $j$ in the optimum and $y_i^\star=1$ if $i\in O$. For $\ell\in\pos$, let $R_\ell^\star=\max_{S\subseteq\calC,|S|=\ell}\sum_{j\in S}\vco_j$, i.e., $\topl{\ell}{\vco}$. All other variables are 0. We show that $(x^\star,y^\star,R^\star)$ is a feasible solution to~\ref{lp:00} with objective value at most $\opt$.
	
	It is easy to check that \eqref{lp:03}, \eqref{lp:04} and \eqref{lp:05} are satisfied. For \eqref{lp:01}, the LHS is the sum of all connection distances that are at least $T_\ell$ in length in the optimum. Since $T_\ell>\vce^\da_\ell$, the LHS is at most the sum of the $\ell$ largest \emph{service costs} and thus, at most $R^\star_\ell$. For \eqref{lp:02}, the maximum LHS is exactly the sum of the $\ell$ largest service costs. Hence all the constraints are satisfied. Finally, we consider the objective value. Using the following simple equation (also see~\cite{chakrabarty2019approximation} and \eqref{eq:conic}), 
	\begin{equation}
		\sum_{\ell\in\pos}(\ww_\ell-\ww_{\nextp{\ell}})\topl{\ell}{\vco}=\sum_{\ell\in[n]}\ww_\ell\cdot\vco^\da_\ell=cost(\ww;\vco),\label{eqn:conic:sparse}
	\end{equation}
	and according to \cref{lem-sparsify}, the objective value is $cost(\ww;\vco)\leq cost(w;\vco)=\opt$.
\end{proof}

We use the same stochastic rounding algorithm as the $\mathrm{Top}\text{-}\ell$ case in the previous section, and \cref{lem-core} still holds. Let $\hat F\subseteq\calF$ be the set of open facilities. We are ready to prove the main result for $\oftm$ with general input vectors.

\begin{theorem}\label{theorem:fault-main}
There exists a randomized polynomial time $666$-approximation for $\oftm$.
\end{theorem}
\begin{proof}
We invoke \cref{lem-core} with $T=T_\ell$, thus $d_{r_j}(j,\hat F)\leq3r_j\dav(j)+435.6T_\ell+X_j$ with probability 1, and $\E[X_j]\leq226.7\sum_{i\in\calF}x_{ij}\calL_{T_\ell}(i,j)$. Because $\topl{\ell}{\vcc}$ pays only for the largest $\ell$ service costs, we have
\begin{align*}
\E[\topl{\ell}{\vcc}]&\leq435.6\ell\cdot T_\ell+\max_{S\subseteq\calC,|S|=\ell}\sum_{j\in S}3r_j\dav(j)+\sum_{j\in\calC}226.7\sum_{i\in\calF}x_{ij}\calL_{T_\ell}(i,j)\\
&\leq229.7R_\ell+435.6\ell\cdot T_\ell,
\end{align*}
where the last inequality is due to the feasibility of $(x,y,R)$ to \ref{lp:00}, hence~\eqref{lp:01} and~\eqref{lp:02}.

Using \eqref{eqn:conic:sparse}, with the output solution $\hat F$, the expected ordered objective of $\vcc$ with respect to $\ww$ is
\begin{align}
&\E\left[\sum_{\ell\in\pos}(\ww_\ell-\ww_{\nextp{\ell}})\topl{\ell}{\vcc}\right]\notag\\
&\leq229.7\sum_{\ell\in\pos}(\ww_\ell-\ww_{\nextp{\ell}})R_\ell+435.6\sum_{\ell\in\pos}(\ww_\ell-\ww_{\nextp{\ell}})\ell T_\ell\notag\\
&\leq229.7\cdot\opt+435.6\sum_{\ell\in\pos}(\ww_\ell-\ww_{\nextp{\ell}})\ell T_\ell,\label{proof:fault-main-1}
\end{align}
where the second inequality is due to~\cref{lemma:feasible:ft} and the feasibility of our solution $(x,y,R)$. Since $T_\ell\leq(1+\epsilon)\vce^\da_\ell+\epsilon \vce^\da_1/n$ with $\vce^\da_\ell$ being the $\ell$-th largest single connection in the optimal solution, it is easy to obtain $\topl{\ell}{\vco}\geq \ell\cdot \vce^\da_\ell$ where $\vco$ is the optimal service cost vector. Therefore,
\begin{align}
&\sum_{\ell\in\pos}(\ww_\ell-\ww_{\nextp{\ell}})\ell T_\ell\notag\\
&\leq(1+\epsilon)\sum_{\ell\in\pos}(\ww_\ell-\ww_{\nextp{\ell}})\ell\vce^\da_\ell+\epsilon\sum_{\ell\in\pos}(\ww_\ell-\ww_{\nextp{\ell}})\ell\cdot \vce^\da_1/n\notag\\
	&\leq(1+\epsilon)\sum_{\ell\in\pos}(\ww_\ell-\ww_{\nextp{\ell}})\topl{\ell}{\vco}+\epsilon\cdot\ww_1\vce^\da_1\leq(1+2\epsilon)\opt.\label{proof:fault-main-2}
\end{align}

Using~\cref{lem-sparsify}, the service cost vector $\vcc$ satisfies
\begin{equation}
\E[cost(w;\vcc)]\leq(1+\delta)\E[cost(\ww;\vcc)]=(1+\delta)\E\left[\sum_{\ell\in\pos}(\ww_\ell-\ww_{\nextp{\ell}})\topl{\ell}{\vcc}\right],\label{proof:fault-main-3}
\end{equation}
thus by \eqref{proof:fault-main-1}\eqref{proof:fault-main-2}\eqref{proof:fault-main-3}, the approximation ratio is $(1+\delta)(229.7+435.6(1+2\epsilon))\leq(665.3+O(\delta+\epsilon))\leq666$, where one chooses $\delta,\epsilon$ that are small enough. The running time of the algorithm is dominated by the enumeration of $\vce^\da_1$ and $\{T_\ell\}_{\ell\in\pos}$, which is polynomial since $\delta$ and $\epsilon$ are constants.
\end{proof}



%% file: appendix.tex
\appendix


\section{Missing Proofs for Robust Ordered $k$-Median}\label{appendix:robust}

\subsection{The Sparse Instance}
We guess $U\in[V^\star,(1+\epsilon)V^\star)$, where $V^\star$ is the minimum objective of integral solutions to \ref{lp:1} when $\lambda=1$. We have $V^\star\leq(1+O(\epsilon))\opt$ using~\cref{theorem:ordered1}. For the optimum $(\fopt, \copt)$, define $c_p^\star=\min_{i\in F^\star}d(p,i),\,\kappa_p^\star=\argmin_{i\in\fopt}d(p,i)$ for $p \in \calF \cup \calC$, and closed balls $\Ball{S}{p,R}=\{i\in S:d(i,p)\leq R\}$. We need the following technical theorems on pre-processing.
\begin{theorem}\label{theorem:sparse}
  (\cite{krishnaswamy2018constant}) Given $\rho,\delta\in(0,1)$ and $U$, there exists an $n^{O(1/\rho)}$-time algorithm that finds an extended instance $\calI'=(\calF,\calC'\subseteq \calC,d,k,m'=|\calC'\cap\copt|,S_0\subseteq \fopt)$ satisfying the following.
  \begin{enumerate}[label=(\ref{theorem:sparse}.\arabic*), ref=(\ref{theorem:sparse}.\arabic*), leftmargin=1.2cm]
\item\label{con:sparse1}
  Denote $\cpt=\copt\cap \calC'$. For each $i \in \fopt\setminus S_0$, we have $\sum_{j \in \cpt:\kappa_j^\star=i} f(c_j^\star) \le \rho U$.
\item\label{con:sparse2}
  For each $p \in \calF \cup \calC'$, we have $\abs{\Ball{\cpt}{p,\delta c_p^\star}} \cdot f((1-\delta)c_p^\star) \le \rho U$.
\item\label{con:sparse3}
    Denote $U'=\sum_{j \in\cpt} f(c_j^\star)$. We have
  $\sum_{j \in \copt\setminus \calC'} f\left(\frac{1-\delta}{1+\delta} d(j,S_0)\right) + U'\le U.$
\end{enumerate}
\end{theorem}
\begin{proof}[Sketch for~\cref{theorem:sparse}] The proof is adapted from~\cite{krishnaswamy2018constant}.
Assume we know $(\fopt,\copt)$ in advance, and we remedy this assumption after our construction. We have $U\in[V^\star,(1+\epsilon)V^\star)$ and $V^\star=\sum_{j\in\copt}f(c^\star_j)$ by definition. Set $S_0=\emptyset,\,\calC'=\calC$ initially. Whenever there exists $i\in\fopt\setminus S_0$ such that $\sum_{j\in\copt:\kappa^\star_j=i}f(c^\star_j)\geq\rho U$, we set $S_0\leftarrow S_0\cup\{i\}$. This process can be repeated at most $O(1/\rho)$ times, because the subsets of clients connected to the facilities in $S_0$ are mutually disjoint, and the corresponding sum of $f$ values is at most $V^\star\leq U$. The remaining facilities in $F^\star\setminus S_0$ will always satisfy \ref{con:sparse1} because we will only add facilities to $S_0$ and remove clients from $\calC'$.

We put $\cpt=\copt\cap\calC'$ at all times. Whenever there exists $p\in\calF\cup\calC'$ such that $\abs{\Ball{\cpt}{p,\delta c_p^\star}} \cdot f((1-\delta)c_p^\star) \geq \rho U$, remove all clients in $\Ball{\calC'}{p,\delta c_p^\star}$ from $\calC'$ and add $\kappa^\star_p$ to $S_0$. All removed clients are from $\calC'$, and using the triangle inequality, each removed $j$ has cost at least $f(c^\star_j)\geq f(c^\star_p-d(j,p))\geq f((1-\delta)c^\star_p)$, thus the total cost removed is at least $\abs{\Ball{\cpt}{p,\delta c_p^\star}} \cdot f((1-\delta)c_p^\star) \geq \rho U$. Using a similar argument, this process can also be repeated at most $O(1/\rho)$ times. The condition in~\ref{con:sparse2} is then satisfied by definition.

Note that each such removed client $j\in\calC\setminus\calC'$ has, by the triangle inequality again,
\[f\left(\frac{1-\delta}{1+\delta}d(j,S_0)\right)\leq f\left(\frac{1-\delta}{1+\delta}(d(j,p)+d(p,\kappa^\star_p))\right)\leq f((1-\delta)c^\star_p)\leq f(c^\star_j),\]
where the last inequality is because $c^\star_j\geq c^\star_p-d(j,p)\geq(1-\delta)c^\star_p$. Therefore by summing over all $j\in\copt$, \ref{con:sparse3} follows since
\[
\sum_{j\in\copt\setminus\calC'}f\left(\frac{1-\delta}{1+\delta}d(j,S_0)\right)+U'\leq
\sum_{j\in\copt\setminus\calC'}f(c^\star_j)+\sum_{j\in\copt\cap\calC'}f(c^\star_j)=V^\star\leq U.
\]

Finally, we remove the dependence of the procedures on $(\fopt,\copt)$, by noticing that $|S_0|=O(1/\rho)$, and $\calC'$ is obtained from $\calC$ by removing $O(1/\rho)$ closed balls. Since $m'$ only takes values in $[m]$, the total number of possible outcomes is at most $n^{O(1/\rho)}$, and we can simply iterate through all possible configurations of $(\calC',m',S_0)$.
\end{proof}


\begin{theorem}\label{theorem:gap}
  Given the desired instance $\calI'$ found in~\cref{theorem:sparse}, we can efficiently compute the set of connection distance upper bounds $\{\widehat{R}_j\geq0: j\in \calC'\}$ satisfying the following.
  \begin{enumerate}[label=(\ref{theorem:gap}.\arabic*), ref=(\ref{theorem:gap}.\arabic*), leftmargin=1.2cm]
  \item\label{property:validity-r}
    There exists a solution $(\fopt, C')$ to $\mathcal{I}''$, s.t. for $j\in C'$ connected to $\kappa_j'\in\fopt,\,c_j'= d(\kappa_j',j)$, one has $c_j'\le (1+3\delta/4)\widehat{R}_j,\,
      \sum_{j \in C'} f\left(\frac{2}{2+\delta}c_j'\right) \le U',\,
      \sum_{j \in C': \kappa_j'=i} f\left(\frac{2}{2+\delta} c_j'\right) \le \rho U,\,\forall i \in \fopt\setminus S_0.$
  \item\label{property:small-backup}
    For each $t > 0,\, p \in \calF \cup \calC'$, one has
    $
      \abs{\qty{j \in \Ball{\calC'}{p, \frac{\delta}{4}t}: \widehat{R}_j \ge t}} \le \frac{\rho U}{f((1-\delta)(1-\delta/4)t)}.$
  \end{enumerate}  
\end{theorem}

\begin{proof}[Sketch for~\cref{theorem:gap}] The proof is adapted from~\cite{krishnaswamy2018constant}.
We iteratively construct $\{\widehat{R}_{j}:j\in\calC'\}$ that always maintain \ref{property:small-backup}, then prove \ref{property:validity-r} using the solution $(\fopt, \cpt=\copt\cap \calC')$ to $\calI'$. Initially let $\widehat{R}_j=0$ for each $j \in \calC'$. In each iteration $k \ge 1$, we try to assign the $k$-th largest distance $t'$ in $\qty{d(i,j):i\in \calF, j\in \calC'}\setminus\{0\}$ to unassigned clients $\{j \in \calC': \widehat{R}_j=0\}$ \emph{sequentially} without violating \ref{property:small-backup}, which is easy to check since it suffices to consider the case of $t=t'$ for each $p\in \calF\cup \calC'$ (also see~\cite{krishnaswamy2018constant}).

We then construct a one-to-one mapping $\phi: \cpt\to \calC'$ and show the solution $(\fopt, \phi(\cpt))$ satisfies \ref{property:validity-r}. Initially, we let $\phi$ be the identity function on $\cpt$.
Consider the clients in $\{j \in \cpt: c_j^\star > (1+3\delta/4)\widehat{R}_j\}$ in non-decreasing order of $c_j^\star$. For each such $j$, we want to update $\phi(j)$ to an unused client in the current $\calC' \setminus \phi(\cpt)$ such that $d(\phi(j),j) \le \delta c_j^\star / 2$ and $\widehat{R}_{\phi(j)} \ge c_j^\star$. If such $\phi(j)$ exists for each $j\in\cpt$, we connect $\phi(j)$ to $\kappa^\star_j$, define $\kappa_{\phi(j)}'=\kappa^\star_j$ and thus $c_{\phi(j)}' = d(\phi(j),\kappa_j^\star) \le c_j^\star + d(j, \phi(j)) \le (1+\delta/2)c_j^\star\le (1+\delta/2)\widehat{R}_{\phi(j)}$. Moreover, one has
\begin{align*}
&c'_{\phi(j)}=d(\phi(j),\kappa_j^\star)\leq(1+\delta/2)c_j^\star,\\
&\Rightarrow\sum_{j\in C':\kappa_j'=i}f\left(\frac{2}{2+\delta}c_j'\right)\leq\sum_{j\in\cpt:\kappa_j^\star=i}f(c_j^\star)\leq\rho U,\,i\notin S_0,
\end{align*}
where the last inequality is due to \ref{con:sparse1}. Similarly, one has
\[\sum_{j\in C'}f\left(\frac{2}{2+\delta}c_j'\right)=\sum_{j\in\cpt}f\left(\frac{2}{2+\delta}c'_{\phi(j)}\right)\leq\sum_{j\in\cpt}f(c_j^\star)=U',\]
therefore \ref{property:validity-r} is satisfied by $(F^\star,\phi(\cpt))$ in this case.

In the following, we show that such an unused $j'\in\calC'\setminus\phi(\cpt)$ can always be found for each $j\in\cpt$ with $c_j^\star\geq(1+3\delta/4)\widehat{R}_j$.
Notice that we have $\widehat{R}_j=0$ when $t'=c_j^\star$ is considered during the construction, so setting $\widehat{R}_j=c_j^\star$ would be a violation and there exists $p \in \calF \cup \calC'$ such that $d(p,j) \le \delta c_j^\star/4$ and the set $H_j=\qty{j'\in \Ball{\calC'}{p,\delta c_j^\star/4}: \widehat{R}_{j'} \ge c_j^\star}$ satisfies $|H_j\cup\{j\}|=|H_j|+1 > \frac{\rho U}{f((1-\delta)(1-\delta/4)c_j^\star)}$. If there exists some $j' \in H_j\setminus \phi(\cpt) $, we can set $\phi(j)=j'$ since $d(j,j')\le \delta c_j^\star/2$ and $\widehat{R}_{j'} \ge c_j^\star$, therefore it suffices to prove $H_j \not\subseteq \phi(\cpt)$.

Assume by contradiction that $H_j \subseteq \phi(\cpt)$ when we consider $j$, then for each $\phi(j') \in H_j$, $d(p, j') \le d(p, \phi(j')) + d(j', \phi(j')) \le 3c_j^\star/4$, using the fact that we consider the clients in non-decreasing order of $c^\star_j$ and hence $d(j',\phi(j'))\leq\delta c_{j'}^\star/2\leq\delta c_{j}^\star/2$ in earlier iterations. This shows that $\phi^{-1}(H_j)\subseteq \Ball{\cpt}{p, 3\delta c_j^\star/4}$. We further notice that $c_p^\star \ge c_j^\star - d(p,j) \ge (1-\delta/4)c_j^\star$ and $\delta c_p^\star\geq \delta(1-\delta/4)c_j^\star \geq(3\delta/4)c_j^\star$ as $\delta < 1$ and thus $\phi^{-1}(H_j)\subseteq \Ball{\cpt}{p, 3\delta c_j^\star/4}\subseteq\Ball{\cpt}{p, \delta c_p^\star}$. Because $j\notin\phi^{-1}(H_j)$ while $j\in\Ball{\cpt}{p,\delta c_p^\star}$, we obtain
\begin{align*}
  \abs{\Ball{\cpt}{p,\delta c_p^\star}}\cdot f((1-\delta)c_p^\star) &\geq (|H_j|+1)f((1-\delta)c_p^\star)\\
  &\ge (|H_j|+1) f((1-\delta)(1-\delta/4)c_j^\star) > \rho U,
\end{align*}
which is a contradiction to \ref{con:sparse2}.
\end{proof}

\subsection{The Strengthened LP}

Let $R_j=(1+3\delta/4)\widehat{R}_j$ in \cref{theorem:gap} and define the following stronger LP relaxation based on~\ref{lp:1}, for $0 < \lambda_1 \le \frac{2}{2+\delta}$. We note that~\ref{lp:ext} is built on the new instance $\calI'$ with clients $\calC'\subseteq\calC$, hence admits a more ``regular'' solution according to~\cref{theorem:gap}. In our algorithm, we solve~\ref{lp:ext} instead of~\ref{lp:1}, and conduct iterative rounding on its solution.
\begin{align*}
  \text{min} && \sum_{j\in\calC'}\sum_{i\in \calF} x_{ij} f(\lambda_1 d(i,j)) \tag{$\mathrm{ExtLP}(f,\lambda_1)$}\label{lp:ext}\\
  \text{s.t.} &&\sum_{j\in \calC'}\sum_{i \in \calF} x_{ij} &\ge m'\\
  && \sum_{i \in \calF} x_{ij} &\le 1 && \forall j \in \calC'\\
  && \sum_{i\in\calF}y_i&=k\\
  && 0\leq x_{ij} &\leq y_i \leq 1  && \forall i \in \calF, j \in \calC'\\
  && y_i & =1 && \forall i \in S_0\tag{$\mathrm{ExtLP}(f,\lambda_1)$.5}\label{lp:ext5}\\
  && x_{ij} & =0 && d(i,j) > R_j\tag{$\mathrm{ExtLP}(f,\lambda_1)$.6}\label{lp:ext6}\\
  && x_{ij} & =0 && \forall i \notin S_0,\, f\left(\frac{2d(i,j)}{2+\delta}\right) > \rho U\tag{$\mathrm{ExtLP}(f,\lambda_1)$.7}\label{lp:ext7}\\
  && \sum_{j} f\left(\frac{2d(i,j)}{2+\delta}\right)x_{ij} &\leq \rho U y_i && \forall i \notin S_0.\tag{$\mathrm{ExtLP}(f,\lambda_1)$.8}\label{lp:ext8}
\end{align*}

\begin{lemma}\label{lemma:extlp}
The optimal objective value of~\ref{lp:ext} is at most $\frac{\lambda_1(2+\delta)}{2}U'$.
\end{lemma}
\begin{proof}
  Using \ref{property:validity-r}, there exists an integral solution having cost at most $U'$ when $\lambda_1=2/(2+\delta)$. For $\lambda_1\leq 2/(2+\delta)$, the same solution is still feasible for the constraints are independent of $\lambda_1$. For $\alpha\leq 1, z>0$, we have $f(\alpha z)\leq \alpha f(z)$ (because for smaller inputs, the coefficient $\hw_t$ is not increasing, see~\eqref{eq:costfunction}), thus 
  \[\sum_{j\in C'}f(\lambda_1 c_j')\leq\frac{\lambda_1(2+\delta)}{2}\sum_{j\in C'}f\left(\frac{2}{2+\delta}c_j'\right)\leq\frac{\lambda_1(2+\delta)}{2}U'.\]
\end{proof}

\begin{lemma}\label{lemma:schedule-rounding}
We can add co-located copies to $\calF$, create a vector $y^\star\in[0,1]^\calF$ and define subsets $F_j\subseteq\Ball{\calF}{j,R_j}$ for each client $j\in\calC'$, such that the following holds.
\begin{enumerate}[label=(\ref{lemma:schedule-rounding}.\arabic*), ref=(\ref{lemma:schedule-rounding}.\arabic*), leftmargin=1.2cm]
\item\label{item:outlier} Outlier constraint: $y^\star(F_j)\leq1$ for each $j\in\calC'$ and $\sum_{j\in\calC'}\big(\sum_{i\in F_j}y^\star_i\big)\geq m'$.
\item\label{item:cardinality} Cardinality constraint: $\sum_{i\in\calF}y^\star_i\leq k$.
\item\label{item:preselected} Pre-selected facilities: for each $i\in S_0$, $\sum_{i'\text{ co-located with }i}y^\star_i=1$.
\item\label{item:total:cost} Bounded objective: $\sum_{j\in\calC'}\sum_{i\in F_j}y^\star_i f(\lambda_1d(i,j))\leq\frac{\lambda_1(2+\delta)}{2}U'$.
\item\label{item:star:cost} Bounded star cost: for each $i$ not co-located with $S_0$, $\sum_{j\in\calC':i\in F_j} f\left(\frac{2}{2+\delta}d(i,j)\right)\le 2\rho U$.
\end{enumerate}
\end{lemma}
\begin{proof}
We start with an optimal solution $(x^\star,y^\star)$ to \ref{lp:ext} with objective at most $\frac{\lambda_1(2+\delta)}{2}U'$, according to \cref{lemma:extlp}. To avoid confusion in notation, we create a copy $\calF'=\calF$, define $F_j=\{i\in\calF':x^\star_{ij}>0\}$ and $\bar y^\star\leftarrow y^\star$ both supported on $\calF'$. For each copy $i'\in\calF'$ of $i\in\calF$, define its star cost $\sum_{j\in\calC':i'\in F_j}f(\frac{2}{2+\delta}d(i,j))$.

We iteratively perform the following procedures. For each $i\in\calF$ and $j\in\calC'$ such that $x^\star_{ij}>0$, we sort all copies of $i$ in $\calF'$ in non-decreasing order of their current star costs, and choose the first several copies such that their $\bar y^\star$ values add up to \emph{exactly} $x^\star_{ij}$. If we need to split a facility $i'$ into two copies to make the sum exact, we replace $i'$ with $\{i_1',i_2'\}$ in $\calF'$, set $\bar y^\star_{i_1'}$ to whichever value is needed and $\bar y^\star_{i_2'}\leftarrow \bar y^\star_{i'}-\bar y^\star_{i_1'}$. Remove from $F_j$ all copies of $i$, and add the selected copies to $F_j$ again. For any other $j'\neq j$, if some $i'\in F_{j'}$ is split in two, $F_{j'}\leftarrow F_{j'}\setminus\{i'\}\cup\{i_1',i_2'\}$.

After the procedures, we set $\calF\leftarrow\calF'$ and the corresponding $y^\star\leftarrow \bar y^\star$, $\{F_j\}_{j\in\calC'}$ such that they are supported on $\calF$. \ref{item:outlier} to \ref{item:total:cost} are easy to see, since the original solution to \ref{lp:ext} is preserved up to facility duplication. To see \ref{item:star:cost}, consider each (original) facility $i$ and all clients $j$ such that $x^\star_{ij}>0$, denoted by $J_i$. It is easy to see that, each copy of $i$ can only appear in $\bigcup_{j\in J_i}F_j$. We use induction to show that, after each iteration, the difference between the maximum and minimum star costs among all copies of $i$ is at most $\rho U$.

The copies of $i$ and their star costs may only change after an iteration where $i$ is selected. Suppose $J_i=\{j_1,\dots,j_\ell\}$ and we consider the iterations in the order of $(i,j_1),\dots,(i,j_\ell)$.
As the base case, before $(i,j_1)$ is considered, the claim is true because $i$ has only one copy in $\calF'$.

Suppose the claim is true after $(i,j_{t-1})$, $t\geq1$. In the start of the iteration where we consider $(i,j_t)$ and sort the copies of $i$ in non-decreasing order of their current star costs, each client $j_{s},s\geq t$ contributes equally to the star cost of each copy of $i$, including $j_t$ in particular, and the difference between the maximum and minimum is at most $\rho U$, using the induction hypothesis. During this iteration, we remove the contributions of $j_t$ to all copies, and add them back to copies that have the smallest star costs. Since $f(\frac{2}{2+\delta}d(i,j_t))\leq\rho U$ by \cref{lp:ext7}, it is easy to verify that the difference between the maximum and minimum after the iteration, is still at most $\rho U$. This finishes the induction.

For facility $i$, we let $\calF(i)\subseteq\calF'$ be the set of copies of $i$ after the procedures. It follows that
\begin{align*}
    \sum_{i'\in\calF(i)}\bar y^\star_{i'}\sum_{j\in\calC':i'\in F_j}f\left(\frac{2}{2+\delta}d(i,j)\right)
    &=\sum_{j\in J_i}f\left(\frac{2}{2+\delta}d(i,j)\right)\sum_{i'\in\calF(i)\cap F_j}\bar y^\star_{i'}\\
    &=\sum_{j\in J_i}x^\star_{ij}f\left(\frac{2}{2+\delta}d(i,j)\right)\leq\rho U y^\star_i,
\end{align*}
where the last inequality is due to \cref{lp:ext8}. The above inequality shows that the minimum star cost is at most $\rho U y^\star_i/\sum_{i'\in\calF(i)}\bar y^\star_{i'}=\rho U$, and thus the maximum star cost is at most $2\rho U$, yielding \ref{item:star:cost}.
\end{proof}


\subsection{Iterative Rounding Details}\label{app:iterative:detail}

We obtain $y^\star\in[0,1]^\calF$ and $\{F_j\}_{j\in\calC'}$ using \cref{lemma:schedule-rounding}. We use a \emph{deterministic} metric discretization method: fix $\tau>1$, define $D_{-2} = -1, D_{-1}=0$ and $D_l = \tau^l$ for $l \ge 0$, and let $d'(i,j)=\min\{D_l \ge d(i,j) : l \ge -2\}$. For each $j\in \calC'$, we call $F_j$ its \emph{outer ball}, define \emph{radius level} $l_j$ such that $D_{l_j} = \max_{i\in F_j} d'(i,j)$, and \emph{inner ball} $B_j=\{i\in F_j: d'(i,j) \le D_{l_j-1}\}$. For $0 < \lambda_2 \le 1/\tau$, we define the auxiliary LP
\begin{align*}
\text{min}&&\sum_{j \in \cpart} \sum_{i \in F_j} y_i f(\lambda_2 d'(i,j)) &+ \sum_{j\in \cfull} \sum_{i \in B_j} y_i f(\lambda_2 d'(i,j))  + (1-y(B_j)) f(\lambda_2 D_{l_j}) \tag{$\mathrm{AuxLP}(f,\lambda_2)$}\label{lp:aux}\\
  \text{s.t.}&&y(F_j)&=1\quad\quad\quad\forall j \in \ccore\tag{$\mathrm{AuxLP}(f,\lambda_2)$.1}\label{lp:aux1}\\
  && 0
  \le y_i&\le 1 \quad\quad\quad\forall i \in \calF\tag{$\mathrm{AuxLP}(f,\lambda_2)$.2}\label{lp:aux2}\\
  && y(B_j)&\le 1 \quad\quad\quad\forall j \in \cfull\tag{$\mathrm{AuxLP}(f,\lambda_2)$.3}\label{lp:aux3}\\
  && y(F_j)&\le 1 \quad\quad\quad\forall j \in \cpart\tag{$\mathrm{AuxLP}(f,\lambda_2)$.4}\label{lp:aux4}\\
  && y(\calF) &\le k\tag{$\mathrm{AuxLP}(f,\lambda_2)$.5}\label{lp:aux5}\\
  && |\cfull| + \sum_{j \in \cpart} y(F_j)&\ge m'.\tag{$\mathrm{AuxLP}(f,\lambda_2)$.6}\label{lp:aux6}
\end{align*}

We then use \cref{iterative-rounding} to find an almost-integral solution with at most 2 fractions.
The algorithm maintains three subsets $\cfull,\cpart$, and $\ccore$, such that $\cfull\cup\cpart=\calC'$, each client in $\cfull$ is to be assigned an open facility relatively close to it, and $\ccore$ is used for placing these facilities. Initially, we set $\cfull=\emptyset$, $\cpart=\calC'$ and $\ccore=S_0$ (each $i \in S_0$ is a \emph{virtual client} and its initial radius level is $-1$).

\begin{algorithm}[ht]
\caption{Iterative Rounding Algorithm~\cite{krishnaswamy2018constant}}\label{iterative-rounding}
\SetKwInOut{Input}{Input}\SetKwInOut{Output}{Output}
\SetKwBlock{Update}{{update-}$\ccore(j)$}{end}
\DontPrintSemicolon
\Input{outer balls $\{F_j:j\in\calC'\}$, radius levels $\{l_j:j\in\calC'\}$, inner balls $\{B_j:j\in\calC'\}$, virtual clients $S_0$ and the input fractional solution $y^\star$}
\Output{an almost-integral solution $y'$ with at most 2 fractions}
$\cfull\leftarrow\emptyset,\cpart\leftarrow\calC',\ccore\leftarrow S_0$\;
\While{true}{
find an optimal basic solution $y'$ to~\ref{lp:aux}\;
\uIf{there exists $j\in\cpart$ such that $y'(F_j)=1$}{
$\cpart\leftarrow\cpart\setminus\{j\},\cfull\leftarrow\cfull\cup\{j\},B_j\leftarrow\{i\in F_j:d'(i,j)\leq D_{l_j-1}\}$, \textbf{update-}$\ccore(j)$\;
}
\uElseIf{there exists $j\in\cfull$ such that $y'(B_j)=1$}{
$l_j\leftarrow l_j-1,F_j\leftarrow B_j,B_j\leftarrow\{i\in F_j:d'(i,j)\leq D_{l_j-1}\}$, \textbf{update-}$\ccore(j)$\;
}
\Else{
break\;
}}
\Return $y'$\;
\Update{
\If{there exists no $j'\in\ccore$ with $l_{j'}\leq l_j$ and $F_j\cap F_j'\neq\emptyset$}{
remove from $\ccore$ all $j'$ such that $F_j\cap F_{j'}\neq \emptyset$,\,$\ccore\leftarrow\ccore\cup\{j\}$\;
}}
\end{algorithm}


\begin{lemma}\label{lemma:iter:feasible}
  In each iteration, $y'$ is feasible after modifying the LP.
  The objective value for $y'$ is non-increasing throughout the algorithm.
\end{lemma}
\begin{proof}
There are two case. The first is when we move some $j$ from $\cpart$ to $\cfull$ when $y'(F_j)=1$. Since $B_j\subseteq F_j$, it satisfies the new constraints in \ref{lp:aux3} and \ref{lp:aux1}, if it is indeed added to $\ccore$. We partition $F_j=B_j\cup(F_j\setminus B_j)$, and the contribution of $j$ to the new objective is the same as when it is in $\cpart$, because each $i\in F_j\setminus B_j$ satisfies $d'(i,j)=D_{l_j}$ by definition.

The second case is when we decrease the radius level $l_j$ and invoke the subroutine on $j\in\cfull$, where $y'(B_j)=1$. Comparing the contributions of $j$ before and after the iteration, they are equal since the old contribution has $1-y'(B_j)=0$, and we can partition the new outer ball $F_j\leftarrow B_j$ in the same way as above.

In both cases, the objective of $y'$ does not change during an iteration, and at the beginning of each iteration, we solve for an optimal basic solution, thus yielding the lemma.
\end{proof}

\begin{property}\label{property:1}
  After each step of iterative rounding, the following properties hold:
  \begin{enumerate}[label=(\ref{property:1}.\arabic*), ref=(\ref{property:1}.\arabic*), leftmargin=1.2cm]
  \item\label{property:11} $\cfull$ and $\cpart$ form a partition of $\calC'$, $S_0 \subseteq \ccore$ and $\ccore \setminus S_0 \subseteq\cfull$.
  \item\label{property:12} $\{F_j: j \in \ccore\}$ are mutually disjoint.
  \item\label{property:13} For each $j \in \calC'$, $D_{l_j} \le \tau R_j$.
  \item\label{property:14} For each $j \in \calC'$, $l_j \ge -1$.
  \item\label{property:15} For each $i$ not co-located with $S_0$, $\sum_{j\in \calC':i\in F_j} f(\frac{2}{2+\delta}d(i,j)) \le 2\rho U$.
  \end{enumerate}
\end{property}
\begin{proof}[Sketch for \cref{property:1}]
The iterative rounding algorithm is directly from \cite{krishnaswamy2018constant}, and we provide a sketch here for completeness. We first notice that by iteratively decreasing the radius levels, no client can obtain a radius level of $-2$, since when $l_j=-1$, its inner ball is $B_j=\{i\in F_j:d'(i,j)\leq D_{-2}=-1\}=\emptyset$, the constraint $y(B_j)\leq 1$ cannot be tight, and we will not invoke the subroutine $\mathbf{update}\textbf{-}\ccore$ on $j$, showing \ref{property:14}.

To see \ref{property:11}, we only need to show that virtual clients in $S_0$ cannot be removed from $\ccore$. From the subroutine $\mathbf{update}\textbf{-}\ccore$, $j'$ can be removed by $j$ only when $l_j<l_{j'}$. But each virtual client starts with a radius level of $-1$, and removing any such virtual client means a radius level of $-2$, contradiction.

\ref{property:12} follows easily from the subroutine. \ref{property:13} is due to $F_j\subseteq\Ball{\calF}{j,R_j}$ at the beginning of iterative rounding, hence $D_{l_j}\leq\max_{i\in F_j}d'(i,j)\leq\tau\max_{i\in F_j}d(i,j)\leq\tau R_j$. 
Lastly, to see \ref{property:15}, we notice that each $F_j$, $j\in\calC'$ is inclusion-wise non-increasing during \cref{iterative-rounding}, therefore the sum in \ref{property:15}, being its star cost, is also non-increasing and at most $2\rho U$, using \ref{item:star:cost}.
\end{proof}

We now establish the connection between~\ref{lp:ext} and~\ref{lp:aux}, making it possible to compare their objectives, before and after the iterative rounding process.

\begin{lemma}\label{lemma:iter:init}
  For each $\lambda_1\in(0,1]$ and $\lambda_2=\lambda_1/\tau$, at the start of the iterative rounding algorithm, the optimal solution $y^\star$ to \ref{lp:ext} is a feasible solution to \ref{lp:aux} with objective value not increased.
\end{lemma}
\begin{proof}
  At the start of iterative rounding, we see $\ccore=S_0$, and because we require $y_i=1$ for $i\in S_0$ in~\ref{lp:ext} and $F_i$ consists of all copies of $i$ (as a virtual client), this constraint is satisfied. Initially, $\cfull$ is empty and we only have $y(F_j)\leq1$ for $j\in \calC'$ and $\sum_{j\in \calC'}y(F_j)\geq m'$, which are true by the feasibility of $y^\star$ to \ref{lp:ext}.
  In~\ref{lp:ext}, each connection between $i,j$ has contribution $x^\star_{ij}f(\lambda_1 d(i,j))$. In \ref{lp:aux}, because we have no full clients in the beginning, its contribution is $y^\star_{i}f(\lambda_2 d'(i,j))$. Since $d'$ is rounded-up by a factor of at most $\tau$ and $\lambda_1=\tau\lambda_2$, we further obtain
  \[y^\star_{i}f(\lambda_2 d'(i,j))\leq y^\star_{i}f(\lambda_1 d(i,j)),\]
  thus the objective of \ref{lp:aux} is at most the objective value of~\ref{lp:ext}.
\end{proof}

\begin{lemma}
  \label{lemma:iter:integral}
  If none of \eqref{lp:aux3} and \eqref{lp:aux4} is tight for some $y'$, then there are at most two fractional variables in $y'$.
  At the conclusion of the algorithm, for each 
  $j \in \cfull$,\[\sum_{i \in \calF: d(i,j) \le \frac{3\tau -1}{\tau-1} D_{l_j}} y_i' \ge 1.\]
\end{lemma}
\begin{proof}
Since $y'$ is an optimal basic solution, if it has $t>0$ strictly fractional variables, there are at least $t$ non-trivial (i.e., not in the form of $y_i\geq0$ or $y_i\leq1$) and independent constraints of \ref{lp:aux} that are tight at $y'$. The remaining constraints form a knapsack constraint \ref{lp:aux6}, and a laminar family \ref{lp:aux1} plus \ref{lp:aux5}, according to \ref{property:12}. The number of such tight independent constraints is therefore at most $t/2+1\geq t$, and we have $t\in\{1,2\}$.

To show the second assertion, we first use induction to show that, for each $j$ that is added to $\ccore$ during \cref{iterative-rounding} with radius level $l$, the final solution satisfies $\sum_{i\in\calF:d(i,j)\leq\frac{\tau+1}{\tau-1}D_l}y'_i\geq1$. The base case is simple for $l=-1$, since we know such $j$ cannot be removed from $\ccore$, and $y'$ satisfies the inequality due to the constraint \ref{lp:aux1}. Suppose the claim is true up to $l-1$, $l\geq0$. For $j$ added to $\ccore$ with radius level $l$, if it is not later removed from $\ccore$, the claim directly follows using \ref{lp:aux1}. Otherwise, if $j$ is later removed by $j'$ with $l_{j'}<l_j=l$ and $F_{j'}\cap F_j\neq\emptyset$, using induction hypothesis, the inequality holds for $j'$ and $l_{j'}$, where $D_{l_{j'}}\leq D_l/\tau$. Using triangle inequality, all these facilities are at a distance at most $\frac{\tau+1}{\tau-1}D_{l_{j'}}+D_{l_{j'}}+D_l\leq(\frac{\tau+1}{\tau(\tau-1)}+\frac{1}{\tau}+1)D_l=\frac{\tau+1}{\tau-1}D_l$ from $j$, showing the induction step.

Back to the proof of the lemma. When we invoke $j$ on its \emph{final} radius $l_j$, if we can indeed add $j$ to $\ccore$, the claim above is sufficient since $\frac{\tau+1}{\tau-1}\leq\frac{3\tau-1}{\tau-1}$. If it cannot be added to $\ccore$, it is because there exists $j'\in\ccore$ with $F_{j'}\cap F_j\neq\emptyset$ and $l_{j'}\leq l_j$. Using the claim on the iteration when we add $j'$ to $\ccore$ with radius level $l_{j'}$, and using triangle inequality, all the facilities in the sum are at a distance at most $\frac{\tau+1}{\tau-1}D_{l_{j'}}+D_{l_{j'}}+D_{l_j}\leq\frac{3\tau-1}{\tau-1}D_{l_j}$, whence the lemma follows.
\end{proof}

\begin{lemma}\label{lemma:iter:final}
  Let $0 < \lambda \le \frac{2\tau-2}{\tau(3\tau-1)(2+\delta)}$, the optimal solution $y^\star$ and outer balls $\{F_j\}_{j\in\calC'}$ obtained from \cref{lemma:schedule-rounding} and the stronger relaxation \ref{lp:ext} where $\lambda_1=\frac{\tau(3\tau-1)}{\tau-1}\lambda$.
  The iterative rounding algorithm on \ref{lp:aux}, $\lambda_2=\frac{3\tau-1}{\tau-1}\lambda$ returns a solution $y'$ with at most two fractions.
  Moreover, $y'$ is a feasible solution to~\ref{lp:1} with objective at most $\frac{\lambda\tau(3\tau-1)(2+\delta)}{2\tau-2}U'$.
\end{lemma}
\begin{proof}
  From \cref{lemma:iter:init} and \cref{lemma:extlp}, $y^\star$ is feasible for \ref{lp:aux} and its objective value is upper bounded by $\frac{\lambda_1(2+\delta)}{2}U'=\frac{\lambda\tau(3\tau-1)(2+\delta)}{2(\tau-1)}U'$. If we further take $y'$ to~\ref{lp:1}, we can connect those in $\cfull\setminus \ccore$ to facilities at most $\frac{3\tau-1}{\tau-1}D_{l_j}$ away according to \cref{lemma:iter:integral}, and the feasibility of $y'$ w.r.t. \ref{lp:1} is guaranteed by Property~\ref{property:1} and \cref{lemma:iter:feasible}. From \cref{lemma:iter:integral}, the final solution $y'$ has at most two fractional values. From \cref{lemma:iter:feasible} and \cref{lemma:iter:init}, the objective value of $y'$ in \ref{lp:1} is upper-bounded by (recall that $d'\geq d$ and $\cfull\cup\cpart$ is a partition of $\calC'$)
  \begin{equation*}
    \sum_{j \in \cpart} \sum_{i \in F_j} y_i' f(\lambda d'(i,j)) + \sum_{j\in \cfull} \sum_{i \in B_j} y_i' f(\lambda d'(i,j))  + (1-y'(B_j)) f\left(\lambda \frac{3\tau-1}{\tau-1}D_{l_j}\right),
  \end{equation*}
  which is at most the objective of \ref{lp:aux} and bounded by $\frac{\lambda\tau(3\tau-1)(2+\delta)}{2\tau-2}U'$.
\end{proof}

The following theorem converts the almost-integral solution $y'$ to an integral one $\hat y$.

\begin{theorem}\label{theorem:rounding}
  There exists $\lambda > 0$ depending on $\delta$ and $\tau$ s.t. we can efficiently compute an integral solution $\hat{y}$ to \ref{lp:1} with objective value at most $5\rho U$ larger than that of $y'$.
\end{theorem}
\begin{proof}
  Suppose $y_{i_1}'$ and $y_{i_2}'$ are two fractional variables and $y_{i_1}'+y_{i_2}'=1$.
  Let $C_1 = \{j \in \cpart: i_1 \in F_j, i_2 \notin F_j\}$ and $C_2 = \{j \in \cpart: i_1 \notin F_j, i_2 \in F_j\}$.
W.l.o.g., we assume $|C_1| \ge |C_2|$.
Our integral solution $\hat{y}$ will be defined as $\hat{y}_{i_1}=1$, $\hat{y}_{i_2}=0$ and $\hat{y}_i=y_i'$ for $i\in \calF\setminus\qty{i_1,i_2}$.
One has $|C_1|+|\cfull| \ge y'_{i_1}|C_1| + y'_{i_2}|C_2| + |\cfull| \ge m'$ using the feasibility of $y'$ to~\ref{lp:aux}.

Using Property~\ref{property:1}, the cost of connecting all of $C_1$ to $i_1$ is at most
  \begin{equation*}
    \label{eq:rounding:1}
    \sum_{j \in C_1} f\left(\frac{2}{2+\delta}d(i_1,j)\right) \le \sum_{j:i_1 \in F_j} f\left(\frac{2}{2+\delta}d(i_1,j)\right) \le 2\rho U.
  \end{equation*}
  It remains to bound cost from connecting full clients close to $i_2$, defined as $J=\qty{j\in \cfull: i_2 \in B_j}$. Let $\gamma > 0$ be some constant, $i^\star=e(i_2,\hat{F})$ and $t'=d(i_2,i^\star)$. Let $J_1=\qty{j \in J: d(j,i_2) > \gamma t'}$ and $J_2=\qty{j \in J: d(j, i_2) \le \gamma t'}$.
  For $j \in J_1$, we have $d(j,i^\star) \le d(j,i_2) + d(i_2,i^\star) < (1+\frac{1}{\gamma}) d(j,i_2)$, thus if we want the following inequality for additional pseudo costs to hold,
  \begin{equation*}
    \label{eq:rounding:2}
    \sum_{j\in J_1} f(\lambda d(j,i^\star)) \le \sum_{j\in J_1} f\left(\frac{2}{2+\delta} d(j,i_2)\right) \le 2\rho U,
  \end{equation*}
  the following choice of parameter suffices,
  \begin{equation}
    \lambda \leq \frac{2}{2+\delta} \cdot \frac{\gamma}{1+\gamma}.\label{eqn:parameter:alpha}
  \end{equation}

  For $j \in J_2$, we have $R_j \ge \frac{D_{l_j}}{\tau} \ge \frac{\tau-1}{\tau(3\tau-1)}(t'-d(j,i_2)) \ge \frac{\tau-1}{\tau(3\tau-1)}(1-\gamma)t'$.
  Let $t=\frac{\tau-1}{\tau(3\tau-1)}(1-\gamma)t'$. Suppose that $\frac{\delta}{4+3\delta}t \ge \gamma t'$, then using \ref{property:small-backup} and recalling $R_j=(1+3\delta/4)\widehat{R}_j$,
  \begin{equation*}
    |J_2| \le \abs{\qty{j \in \Ball{\calC'}{i_2, \frac{\delta}{4+3\delta}t}: R_j \ge t}} \le \frac{\rho U}{f(\frac{(1-\delta)(1-\delta/4)}{1+3\delta/4} t)}.
  \end{equation*}
  Using the triangle inequality, we have $d(j,i^\star)\leq(1+\gamma)t'$ and the following total cost of connecting $J_2$ to $i^\star$,
  \begin{equation*}
    \label{eq:rounding:3}
    \sum_{j\in J_2} f(\lambda d(j,i^\star)) \le f(\lambda (1+\gamma) t') |J_2| \le  \rho U,
  \end{equation*}
  where one could choose the following value,
  \begin{equation}
    \lambda \leq\frac{(1-\delta)(1-\delta/4)}{(1+3\delta/4)} \cdot \frac{\tau-1}{\tau(3\tau-1)} \cdot \frac{1-\gamma}{1+\gamma}.\label{eqn:parameter:beta}
  \end{equation}
  
  Denote $\sigma=\frac{\tau-1}{\tau(3\tau-1)}$ and let $\gamma = \frac{\delta \sigma}{4+3\delta+\delta \sigma}$ so that $\frac{\delta}{4+3\delta}t = \gamma t'$.
  Letting $\lambda$ be the minimum of \eqref{eqn:parameter:alpha} and \eqref{eqn:parameter:beta} and summing over the three cases, the increase of objective value w.r.t. \ref{lp:1} is at most $2\rho U + 2\rho U + \rho U = 5\rho U$, thus the theorem follows.
\end{proof}

\subsection{Proof of \cref{theorem:robust-main}}

\begin{proof}[Proof of~\cref{theorem:robust-main}]
Let $\frac{\tau-1}{\tau(3\tau-1)} = 0.101,\,\delta=0.81765$, thus having $\lambda\approx0.008856$.
Fix $\epsilon>0$, obtain a faithful function $f$ using \cref{theorem:ordered1}. Fix $\delta,\rho>0$, compute $\calC',m',S_0,\{R_j:j \in \calC'\}$ via \cref{theorem:sparse} and \cref{theorem:gap}, and solve the stronger \ref{lp:ext} with $\lambda_1=\frac{\tau(3\tau-1)}{\tau-1}\lambda$. 
Using iterative rounding, we obtain an almost-integral solution to \ref{lp:1} using \cref{lemma:iter:final}. Next, we compute an integral solution $\hat y$ using \cref{theorem:rounding}, with the objective w.r.t. \ref{lp:1} increased by at most $5\rho U$. Let $\hat F=\{i\in\calF:\hat y_i=1\}$, greedily connect $m-m'$ clients in $\calC\setminus\calC'$ minimizing $f\left(\frac{1-\delta}{1+\delta} d(j,\hat F)\right)$ for each of them
and output the final solution $(\hat F,\hat C)$. We consider the pseudo costs of $\hat C\setminus\calC'$ and $\hat C\cap\calC'$ separately,
\begin{align}
  &\sum_{j\in\hat C\setminus \calC'} f(\lambda d(j,\hat F))+\sum_{j\in\hat{C}\cap\calC'}f(\lambda d(j,\hat{F}))\notag\\
   \leq &\sum_{j\in\copt\setminus \calC'}\frac{\lambda(1+\delta)}{1-\delta} f\left(\frac{1-\delta}{1+\delta} d(j,\hat F)\right) + \left(\frac{\lambda\tau(3\tau-1)(2+\delta)}{2\tau-2}U' + 5\rho U\right)\notag\\
   \leq &\max\left\{\frac{\lambda(1+\delta)}{1-\delta},\frac{\lambda\tau(3\tau-1)(2+\delta)}{2\tau-2} \right\}\left(\sum_{j\in \copt\setminus\calC'}f\left(\frac{1-\delta}{1+\delta} d(j,S_0)\right)+U'\right)+5\rho U\notag\\
  \leq & \:0.12354U+5\rho U.\label{proof:robust-main:1}
\end{align}
In the above, the first inequality is due to~\cref{lemma:iter:final}, \cref{theorem:rounding} and the greedy selection of $\hat C\setminus\calC'$. The second one is because $S_0\subseteq\hat F$. The last one is according to \ref{con:sparse3} and the choices of parameters. 

Recall that $U\le (1+\epsilon)V^\star$ where $V^\star$ is the optimal objective for \ref{lp:1} when $\lambda=1$, so using \cref{theorem:ordered1}, one has $U\leq(1+O(\epsilon))\opt$. Using~\eqref{proof:robust-main:1} and~\cref{theorem:ordered1} again, the overall cost of the solution $(\hat F,\hat C)$ to the original problem is $(0.12354U+5\rho U+(1+O(\epsilon))\opt)/\lambda\leq (126.9+O(\epsilon+\rho))\opt\leq127\opt$, by choosing small enough $\rho,\epsilon$. The running time is obtained from the enumeration process and thus bounded by a polynomial.
\end{proof}


\section{Ordered Matroid Median}\label{appendix:matroid}

The \emph{matroid median} problem~\cite{krishnaswamy2011matroid,krishnaswamy2018constant,swamy2016improved} is a generalization of $k$-median with the cardinality constraint $k$ replaced by a matroid $\calM$ with rank function $r$, supported on $\calF$, and we are required to open facilities that form an independent set of $\calM$. Here we consider the natural generalization of matroid median, $\omm$, with the objective changed to the ordered objective, given the non-increasing non-negative vector $w\in\mathbb{R}^{|\calC|}$. We give the first constant approximation that we are aware of.

As is pointed out by~\cite{krishnaswamy2018constant}, the natural LP relaxation for matroid median has a small integrality gap, therefore we skip the pre-processing steps and provide a sketch on how the iterative rounding framework gives rise to the algorithm. Our pseudo cost function $f$ is defined exactly as before.
\begin{enumerate}
    \item Ignore the pre-processing steps (obtaining sparse instances and constructing $\widehat{R}_j$) in the $\oko$ algorithm, and replace the cardinality constraint $\sum_{i\in\calF}y_i=k$ with $\sum_{i\in S}y_i\leq r_\calM(S),\,\forall S\subseteq\calF$, following a classic matroid characterization result by Edmonds \cite{edmonds2001submodular}. We proceed similarly to the auxiliary LP for iterative rounding.
    \item Because we do not have any outliers and each client will end up in $\cfull$, the remaining \emph{tight} constraints after the termination of iterative rounding, consist of only a partition matroid ($y(F_j)=1$ for $j\in \ccore$) and a subset of matroid constraints. Therefore the corresponding output solution $y'$ is integral.
    \item Using the same argument as~\cref{lemma:iter:final}, the objective value of $y'$ in the original relaxation \ref{lp:1} is at most bounded by \ref{lp:aux} where $\lambda_2=\frac{3\tau-1}{\tau-1}\lambda$ and \ref{lp:ext} where $\lambda_1=\frac{\tau(3\tau-1)}{\tau-1}\lambda$. Similar to \cref{lemma:extlp}, the objective value of \ref{lp:ext} is at most $\lambda_1(1+\epsilon)\opt$, hence using~\cref{theorem:ordered1}, the final solution has an approximation ratio at most
    \[\min_\lambda\left\{\frac{1}{\lambda}\left(\frac{\tau(3\tau-1)}{\tau-1}\lambda(1+\epsilon)+(1+O(\epsilon))\right)\right\}=10+4\sqrt{6}+O(\epsilon)\leq19.798+O(\epsilon)\leq19.8,\]
    where we choose small enough $\epsilon$, and $\lambda\leq\frac{\tau-1}{\tau(3\tau-1)}\in(0,5-2\sqrt{6}]$ due to the constraint $\frac{\tau(3\tau-1)}{\tau-1}\lambda=\lambda_1\leq1$ from the construction of \ref{lp:ext}.
\end{enumerate}

\section{Ordered Knapsack Median}\label{appendix:knapsack}

The \emph{knapsack median} problem~\cite{byrka2015knapsack,krishnaswamy2018constant,kumar2012constant,swamy2016improved} is also a generalization of $k$-median, by replacing the cardinality constraint $k$ with a given knapsack constraint $W$. Formally, for each facility $i\in\calF$ there is a non-negative weight $\wt_i\geq 0$, and we are required to select a subset $F\subseteq \calF$ of facilities such that their total weight is at most $W$, i.e., $\wt(F)=\sum_{i\in F}\wt_i\leq W$. Here we consider the natural generalization of knapsack median, $\okm$, with only the objective changed to the ordered version, given the non-increasing non-negative vector $w\in\mathbb{R}^{|\calC|}$. We give the first constant approximation that we are aware of.

We follow the procedures given in~\cite{krishnaswamy2018constant}, use an iterative rounding process similar to $\oko$, and use the same pseudo cost function $f$ as before.

\begin{theorem}
  \label{theorem:knap:sparse}
  Given $\rho,\delta\in(0,1)$ and upper bound $U$ on the cost of the optimal solution $\fopt$, there exists an $n^{O(1/\rho)}$-time algorithm that finds an extended instance $\calI'=(\calF,\calC'\subseteq \calC,d,\wt,W,S_0\subseteq \fopt)$ that satisfies
    \begin{enumerate}[label=(\ref{theorem:knap:sparse}.\arabic*), ref=(\ref{theorem:knap:sparse}.\arabic*), leftmargin=1.2cm]
\item\label{con:knap:sparse1}
  For each $i \in \fopt\setminus S_0$, we have $\sum_{j \in \calC,\kappa_j^\star=i} f(c_j^\star) \le \rho U$,
\item\label{con:knap:sparse2}
  For each $p \in \calF \cup \calC'$, we have $\abs{\Ball{\calC}{p,\delta c_p^\star}} \cdot f((1-\delta)c_p^\star) < \rho U$,
\item\label{con:knap:sparse3}
  $\sum_{j \in \calC\setminus\calC'} f\left(\frac{1-\delta}{1+\delta} d(j,S_0)\right) + \sum_{j \in\calC'} f(c_j^\star) \le U.$
\end{enumerate}
\end{theorem}

\begin{theorem}
  \label{theorem:knap:gap}
  Given the instance found in~\cref{theorem:knap:sparse}, we can efficiently compute the set of connection distance upper bounds $\{R_j\geq0: j\in \calC'\}$ such that
   \begin{enumerate}[label=(\ref{theorem:knap:gap}.\arabic*), ref=(\ref{theorem:knap:gap}.\arabic*), leftmargin=1.2cm]
    \item\label{property:knap:small-backup}
    For each $t > 0,\, p \in \calF \cup \calC'$, we have
    \begin{equation*}
      R_j=\max\left\{R>0:\abs{\qty{j \in \Ball{\calC'}{j, \delta R}}}\cdot f((1-\delta)R) \le \rho U\right\}.
    \end{equation*}
  \item\label{property:knap:validity-r}
    For any $j\in\calC'$ we have $c^\star_j\leq R_j$.
    Moreover, for each $i \in \fopt\setminus S_0$, we have
    \begin{equation*}
      \sum_{j \in \calC': \kappa_j^\star=i} f\left(c_j^\star\right) \le \rho U.
    \end{equation*}
  \end{enumerate}  
\end{theorem}

The two theorems above are almost identical to those for $\oko$ and those in Section~7 of~\cite{krishnaswamy2018constant}, thus we omit their proofs here. By replacing the cardinality constraint $y(F)\leq k$ with the relaxed knapsack constraint $\sum_{i\in\calF}\wt_i\cdot y_i\leq W$, and removing the coverage constraint for outliers, we proceed with the stronger LP defined similar to~\ref{lp:ext}, and also iterative rounding on an auxiliary LP similar to~\ref{lp:aux}. Using a similar argument~\cite{krishnaswamy2018constant}, we see that after the termination of iterative rounding, the resulted solution $y'$ corresponds to the intersection of a laminar family and a knapsack polytope, hence it contains at most 2 fractional variables. 

There is no scaling of $2/(2+\delta)$ in~\cref{theorem:knap:gap}, so the comparison of objective values between different relaxations is simpler. We now focus on obtaining the integral solution $\hat{y}$ from $y'$.

\begin{theorem}
  \label{theorem:knap:rounding}
  There exists $\lambda > 0$ depending on $\delta$ and $\tau$, such that we can efficiently compute an integral solution $\hat{y}$ to \ref{lp:1} (in the knapsack case), and its objective value is at most $3\rho U$ larger than that of $y'$.
\end{theorem}
\begin{proof}
If there is only one fractional facility $i_2$, we close it.
  If there are two, suppose $i_1,i_2$ are the two fractional facilities and $i_1$ is the one with a smaller weight. We fully open $i_1$ and close $i_2$.

  Unlike $\oko$, each client is fully connected, so
  it only remains to bound the cost from connecting clients that are close to $i_2$: $J=\qty{j\in \cfull: i_2 \in B_j}$. Let $\gamma > 0$, $i^\star=e(i_2,\hat{F})$ and $t'=d(i_2,i^\star)$. Let $J_1=\qty{j \in J: d(j,i_2) > \gamma t'}$ and $J_2=\qty{j \in J: d(j, i_2) \le \gamma t'}$.
  For $j \in J_1$, we have $d(j,i^\star) \le d(j,i_2) + d(i_2,i^\star) < (1+\frac{1}{\gamma}) d(j,i_2)$, thus
  \begin{equation*}
    \label{eq:rounding:knap2}
    \sum_{j\in J_1} f(\lambda d(j,i^\star)) \le \sum_{j\in J_1} f\left(d(j,i_2)\right) \le 2\rho U,
  \end{equation*}
  where
  \begin{equation}
    \lambda \leq \frac{\gamma}{1+\gamma}.\label{eqn:knapsack:alpha}
  \end{equation}

  Fix some $j \in J_2$, we have $R_j \ge \frac{D_{l_j}}{\tau} \ge \frac{\tau-1}{\tau(3\tau-1)}(t'-d(j,i_2)) \ge \frac{\tau-1}{\tau(3\tau-1)}(1-\gamma)t'$. Suppose that $\delta R_j\ge 2\gamma t'$, then using \ref{property:knap:small-backup},
  \begin{equation*}
    |J_2| \le\abs{\Ball{\calC'}{i_2,\gamma t'}}\le \abs{\Ball{\calC'}{j, 2\gamma t'}} \le \abs{\Ball{\calC'}{j, \delta R_j}} \le \frac{\rho U}{f((1-\delta)R_j)}.
  \end{equation*}
  Using the triangle inequality, we have $d(j,i^\star)\leq(1+\gamma)t'$ and the following total cost of connecting $J_2$ to $i^\star$,
  \begin{equation*}
    \label{eq:rounding:knap3}
    \sum_{j\in J_2} f(\lambda d(j,i^\star)) \le f(\lambda (1+\gamma) t') |J_2| \le  \rho U,
  \end{equation*}
  where
  \begin{equation}
    \lambda\leq (1-\delta)\cdot\frac{\tau-1}{\tau(3\tau-1)} \cdot \frac{1-\gamma}{1+\gamma}.\label{eqn:knapsack:beta}
  \end{equation}
  
  Here we denote $\sigma=\frac{\tau-1}{\tau(3\tau-1)}$ and let $\gamma = \frac{\delta \sigma}{2+\delta \sigma}$ so that $\delta R_j\geq2\gamma t'$.
  Letting $\lambda$ be the minimum of \eqref{eqn:knapsack:alpha} and \eqref{eqn:knapsack:beta} and summing over the two cases, the increase of objective value w.r.t. \ref{lp:1} is at most $2\rho U + \rho U = 3\rho U$, thus the theorem follows.
\end{proof}

Let $\delta=2/3$ and thus $\lambda=\frac{\sigma}{3+2\sigma}$.
The objective value of \ref{lp:ext}, $\lambda_1=\frac{\tau(3\tau-1)}{\tau-1}\lambda$ is at most $\frac{\tau(3\tau-1)}{\tau-1}\lambda(1+\epsilon)\opt$.
Using the same argument as~\cref{lemma:iter:final}, the objective value of $y'$ in the original relaxation \ref{lp:1} is at most \ref{lp:aux}, $\lambda_2=\frac{3\tau-1}{\tau-1}\lambda$. Using \cref{theorem:knap:rounding}, the objective of $\hat y$ to \ref{lp:1} is at most $\frac{\tau(3\tau-1)}{\tau-1}\lambda(1+\epsilon)\opt+3\rho U$.
Finally, using an argument similar to~\eqref{proof:robust-main:1}, the approximation ratio is
\begin{align*}
    &\left(\max\left\{5\lambda,\frac{\tau(3\tau-1)}{\tau-1}\lambda\right\}+1+O(\rho)\right)\frac{1+\epsilon}{\lambda}\\
    &\leq\left(\max\left\{\frac{5\sigma}{3+2\sigma},\frac{1}{3+2\sigma}\right\}+1+O(\rho)\right)\frac{1+\epsilon}{\lambda}.
\end{align*}
Letting $\tau=1+\sqrt{\frac{2}{3}},\,\sigma=5-2\sqrt{6}$, the approximation ratio evaluates to $22+8\sqrt{6}+O(\epsilon+\rho)\leq 41.596+O(\epsilon+\rho)\leq41.6$, where one chooses $\epsilon,\rho$ that are small enough.

\section{Missing Proofs for Fault-Tolerant Ordered $k$-Median}\label{appendix:fault}

\subsection{Stochastic Rounding}

\paragraph{Constructing Bundles and the Laminar Family.}
We use the rounding procedure in~\cite{hajiaghayi2016constant}. Fix an optimal fractional solution $(x,y,R)$ to~\ref{lp:00}, we create a family $\calU$ of non-intersecting sets of facility locations called \emph{bundles}. In what follows, we assume $x_{ij}\in\{0,y_i\}$ for each $i\in\calF,\,j\in\calC$ using standard duplication of facilities (see, e.g.,~\cite{guha2003constant,hajiaghayi2016constant}).

Let $\calF_j=\{i\in\calF:x_{ij}>0\}$ and denote $y(S)=\sum_{i\in S}y_i$ the \emph{volume} of $S$, so $y(\calF_j)=r_j$. There exists a partition of $\calF_j=\calF_{j,1}\cup\cdots\cup\calF_{j,r_j}$ such that $\calF_{j,p}$ contains the $p$-th closest unit volume to $j$ in $\calF_j$. Define $\dav^{p}(j)=\sum_{i\in\calF_{j,p}}x_{ij}d(i,j)$ and $\dmax^p(j)=\max_{i\in\calF_{j,p}}d(i,j),\dmin^p(j)=\min_{i\in\calF_{j,p}}d(i,j)$. Define the overall average $\dav(j)=r_j^{-1}\sum_{i\in\calF_j}x_{ij}d(i,j)$, and it is easy to see $r_j\dav(j)=\sum_{p=1}^{r_j}\dav^p(j)$. Reload the function $\dmax(j,S)=\max_{i\in S}d(i,j)$ as the maximum distance between $j$ and a set of facilities $S$.
Also define similar measures for the truncated distance function $\calL_T$. Define the closed ball of facilities $\ball(j,R)=\{i\in\calF:d(i,j)\leq R\}$, $R\geq0$.
For the sake of being self-contained, we provide the algorithm~\cite{yan2015lp} as follows.

\begin{algorithm}[ht]
\caption{Bundle Creation}\label{bundle-creation}
\SetKwInOut{Input}{Input}\SetKwInOut{Output}{Output}
\DontPrintSemicolon
\Input{a feasible solution $(x,y,R)$ to~\ref{lp:00}, $\{\calF_j:j\in\calC\}$}
\Output{a set of bundles and a subset of bundles for each client}
$\forall j\in\calC,\que_j\leftarrow\emptyset,\calF_j'\leftarrow\calF_j;\,\calU\leftarrow\emptyset$\;
\While{there exists $j$ s.t. $|\que_j|<r_j$}{
choose such $j$ s.t. if $U$ is the closest unit volume in $\calF_j'$, $\dmax(j,U)$ is minimized\;
\eIf{there exists $U'\in\calU$ and $U\cap U'\neq\emptyset$}{
append $U'$ to the end of $\que_j$, $\calF_j'\leftarrow \calF_j'\setminus U'$\;
}{
append $U$ to the end of $\que_j$, $\calU\leftarrow\calU\cup\{U\}$, $\calF_j'\leftarrow\calF_j'\setminus U$\;
}}
\Return{$\calU,\{\que_j:j\in\calC\}$}
\end{algorithm}

We proceed to create another laminar family indexed by a certain subset of clients in $\calC$.
We call a client $j\in\calC$ \emph{dangerous} if $\dmax^{r_j}(j)>45\dav^{r_j}(j)$, and denote the set of dangerous clients $D\subseteq\calC$. We call two dangerous clients $j,\,j'$ \emph{in conflict} if $r_j=r_{j'}$ and $d(j,j')\leq6\max\{\dav^{r_j}(j),\dav^{r_{j'}}(j')\}$. This definition of conflict is somewhat different from~\cite{hajiaghayi2016constant}, since we need to find a subset of dangerous clients that are far enough apart from each other. Using the filtering process in~\cite{hajiaghayi2016constant}, we create $D'\subseteq D$ such that no two clients in $D'$ are in conflict: $D'\leftarrow\emptyset$, in non-increasing order of $\dav(j)$ s.t. $j$ is \emph{unmarked}, $D'\leftarrow D'\cup\{j\}$ and \emph{mark each unmarked} client in $D$ that is in conflict with $j$. The laminar family shall be indexed by $D'$.

We define the closed balls $B_j=\ball(j,\dmax(j)/15),\,j\in D'$, where we abbreviate $\dmax^{r_j}(j)$ to $\dmax(j)$. It is easy to see that $y(B_j)<r_j$ and the set $\calF_j\setminus B_j$ is fully contained in $\calF_{j,r_j}$. We present the simple~\cref{laminar-constructing} used in~\cite{hajiaghayi2016constant} that constructs a laminar family $\calB$. The following lemma follows from~\cite{hajiaghayi2016constant}, thus we omit the proof here (indeed, though our definition of conflict is slightly changed, the proofs are still valid).

\begin{lemma}\label{lem-laminar}(\cite{hajiaghayi2016constant})
$B_j'\subseteq\ball(j,\dmax(j)/10)$ for each $j'\in D'$.
$\calB=\{B_j':j\in D'\}$ is laminar.
\end{lemma}

\begin{algorithm}[ht]
\caption{Laminar Construction}\label{laminar-constructing}
\SetKwInOut{Input}{Input}\SetKwInOut{Output}{Output}
\Input{$D'$ of non-conflicting clients, $\{B_j:j\in D'\}$}
\Output{a laminar family of facilities indexed by $D'$}
\For{$j\in D'$ in non-decreasing order of $r_j$}{
let $D''$ be the set of clients $j'\in D'$ that satisfies $r_{j'}<r_j$ and $B_{j'}'\cap B_j\neq\emptyset$\;
$B_j'\leftarrow B_j\cup\left(\bigcup_{j'\in D''}B_{j'}'\right)$\;}
\Return{$\calB=\{B_j':j\in D'\}$}
\end{algorithm}

\paragraph{Rounding the LP Solution.}
Recall that in the last section, we duplicate some facilities into co-located copies. We abuse the notation and denote $\calF$ the set after duplication, and $\calF_0$ the original set of facility locations. We define $g:\calF\rightarrow\calF_0$ by taking a copy to its original location in $\calF_0$. The following auxiliary LP~\cite{hajiaghayi2016constant} is defined by two laminar families: $\calU$ and $\calB\cup\{\calF\}\cup\{g^{-1}(i'):i'\in\calF_0\}$, thus it is integral.
\[\begin{aligned}
	\sum_{i\in U}z_i&=1 & \forall U\in\calU\\
	\sum_{i\in B_j'}z_i&\in[r_j-1,r_j] & \forall j\in D'\end{aligned}
	\quad\quad\quad
	\begin{aligned}
	\sum_{i\in g^{-1}(i')}z_i&\leq1 & \forall i'\in\calF_0\\
	\sum_{i\in\calF}z_i&=k. &
	\end{aligned}
\]

Since $y$ is a feasible solution, it is the convex combination of a polynomial number of integral solutions using Carath\'{e}odory's Theorem. We can efficiently and randomly sample an integral solution $z'$ such that $\E[z']=y$. It means that
\begin{itemize}
	\item For each $j\in D'$, $\Pr[z'(B_j')=r_j]=y(B_j')-(r_j-1),\,\Pr[z'(B_j')=r_j-1]=r_j-y(B_j')$.
	\item For each $i'\in\calF_0$, the probability of its opening is $y_{i'}$.
\end{itemize}

\subsection{Proof of~\cref{lem-core}}\label{app:lem-core}

In~\cref{bundle-creation}, it is easy to see for each $j\in\calC$, the queue $\que_j$ has size exactly $r_j$ and contains distinct bundles in $\calU$, because whenever $U$ is added to $\que_j$, $U$ is also removed from $\calF_j'$. Let $\que_j=(U_{j,p})_{p=1}^{r_j}$, where $U_{j,p}$ is the $p$-th bundle added. We first need the following lemma.

\begin{lemma}\label{lem-bundle}
For any $p\in[r_j]$, $\dmax(j,U_{j,p})\leq3\dmax^p(j)$.
\end{lemma}
\begin{proof}
	We use induction on the size of $\que_j$. When $U_{j,1}$ is added to the queue, either $U_{j,1}=\calF_{j,1}$ or $U_{j,1}=U_{j',q}$ which is previously created from some $\calF_{j'}',j'\neq j$ satisfying $\dmax(j',U_{j',q})\leq\dmax(j,\calF_{j,1})$ and $U_{j',q}\cap\calF_{j,1}\neq\emptyset$. In the first case, the claim is trivial. In the second case, using the triangle inequality we have
	\begin{align*}
	    \dmax(j,U_{j',q})&\leq d(j,j')+\dmax(j',U_{j',q})\leq\dmax(j,\calF_{j,1})+2\dmax(j',U_{j',q})\\
	    &\leq 3\dmax(j,\calF_{j,1})=3\dmax^1(j).
	\end{align*}
	Assume the claim holds for up to $p-1$. For $U_{j,p}$, if it is created from $\calF_{j}'$, in other words, the closest unit-volume set $V_{j,p}\subseteq\calF_{j}'$ that intersects with no other bundles in $\calU$, since we have only added $p-1$ bundles to $\que_j$, and with each addition, we remove at most 1 volume of facilities from $\calF_j'$, the total volume removed from $\calF_j'$ is at most $p-1$, hence it follows that $\dmax(j,V_{j,p})\leq \dmax^p(j)$.
	
	If $U_{j,p}=U_{j',q}$ is created from $\calF_{j'}'$ satisfying $\dmax(j',U_{j',q})\leq\dmax(j,V_{j,p})$ and $U_{j',q}\cap V_{j,p}\neq\emptyset$, we use the triangle inequality again and obtain
	\[\dmax(j,U_{j',q})\leq d(j,j')+\dmax(j',U_{j',q})\leq\dmax(j,V_{j,p})+2\dmax(j',U_{j',q})\leq3\dmax^p(j).\]
\end{proof}

\begin{proof}[Proof of~\cref{lem-core}]
We prove this lemma for $\calC\setminus D,\,D'$ and $D\setminus D'$ separately. For each $j$, we associate one deterministic cost value $D_j$ and another random variable cost $X_j$ to it. Initially, $D_j=0$ and $X_j=0$ with probability 1. Our goal is to satisfy $\Pr[d_{r_j}(j,\hat F)\leq D_j+X_j]=1$. In the following, we define $\davt{T}^p(j)=\sum_{i\in\calF_{j,p}}x_{ij}\calL_T(i,j),\,p\in[r_j]$, i.e., the weighted average of $\calL_T(\cdot,j)$ on the $p$-th partition $\calF_{j,p}\subseteq\calF_j$.

\paragraph{{(I)} For $j\in\calC\setminus D$.}\label{proof:667-1}
Consider $\que_j=\{U_{j,p}\}_{p=1}^{r_j}$. All bundles in $\que_j$ are disjoint, and our rounding guarantees one open facility in each bundle. By connecting $j$ to these $r_j$ open facilities, we obtain
	\[d_{r_j}(j,\hat F)\leq\sum_{p=1}^{r_j}3\dmax^p(j)\leq 3\dmax^{r_j}(j)+3\sum_{p=2}^{r_j}\dav^p(j),\]
	with probability 1, where the first inequality is due to~\cref{lem-bundle} and the second inequality is due to $\dmax^p(j)\leq\dav^{p+1}(j)$. The second sum above is at most $3r_j\dav(j)$. Further, since $j\in\calC\setminus D$, it is not dangerous and thus $\dmax^{r_j}(j)\leq45\dav^{r_j}(j)$. The service cost can be bounded by
	\[d_{r_j}(j,\hat F)\leq3r_j\dav(j)+135\dav^{r_j}(j)\leq3r_j\dav(j)+135T+135\davt{T}^{r_j}(j),\]
	where the last inequality is due to the simple observation using $\sum_{i\in\calF_{j,r_j}}x_{ij}=1$, \[\dav^{r_j}(j)=\sum_{i\in\calF_{j,r_j}}x_{ij}d(i,j)\leq\sum_{i\in\calF_{j,r_j}}x_{ij}(\calL_T(i,j)+T)\leq T+\davt{T}^{r_j}(j).\] 
	We charge $D_j=3r_j\dav(j)+135T$, and charge a fixed value $X_j=135\davt{T}^{r_j}(j)$ with probability 1.

\paragraph{{(II)} For $j\in D'$.}\label{proof:667-2}
We consider the following cases, in which we charge $D_j$ and $X_j$ differently, that is, define the two variables according to the different scenarios.

\subparagraph{(A) $T\geq\dmax(j)/15$.} The first $r_j-1$ assignments can be satisfied by the distinct open facilities in $U_{j,1},\dots,U_{j,r_j-1}$, with a total service cost at most
$3\dmax^1(j)+\cdots+3\dmax^{r_j-1}(j)\leq3r_j\dav(j)$ using \cref{lem-bundle}. We charge $3r_j\dav(j)$ to $D_j$. The $r_j$-th closest open facility is at most $3\dmax(j)\leq45T$ away in $U_{j,r_j}$, and we simply charge $45T$ more to $D_j$, thus $D_j=3r_j\dav(j)+45T,\,X_j=0$.

\subparagraph{(B) $T<\dmin^{r_j}(j)$.} Notice that $\dav^{r_j}(j)=\davt{T}^{r_j}(j)$ because $d(i,j)\geq\dmin^{r_j}(j)>T$ for each $i\in\calF_{j,r_j}$. 
Define the following two random variables: $A_j$ is the total cost of connecting $j$ to \emph{all} open facilities in $B_j'\cap\calF_{j,r_j}$, and $Q_j=3\dmax(j)$ when there are only $r_j-1$ open facilities in $B_j'$ and 0 otherwise. It is easy to obtain the following, using the marginal distribution of $\hat F$,
\begin{align*}
\E[A_j]&=\sum_{i\in B_j'\cap\calF_{j,r_j}}x_{ij}d(i,j)\leq\sum_{i\in\calF_{j,r_j}}x_{ij}d(i,j)\leq\dav^{r_j}(j)=\davt{T}^{r_j}(j),\\
\E[Q_j]&=3\dmax(j)(r_j-y(B_j'))\leq45\dav^{r_j}(j)=45\davt{T}^{r_j}(j).
\end{align*}
In this case, we charge $D_j=3r_j\dav(j)$ and $X_j=A_j+Q_j$, therefore $\E[X_j]\leq46\davt{T}^{r_j}(j)$. We show that $d_{r_j}(j,\hat F)\leq D_j+X_j$ with probability 1 in the following cases.
\begin{enumerate}
\item When there are $r_j$ open facilities in $\calF_j\setminus \calF_{j,r_j}$, we first connect $j$ to the open facilities in bundles $U_{j,1},\dots,U_{j,r_j-2}$. There are at least 2 more in $\calF_j\setminus \calF_{j,r_j}$, both with a connection cost at most $\dmax^{r_j-1}(j)\leq\dav^{r_j}(j)$. In this case, we have $d_{r_j}(j,\hat F)\leq3(\dmax^1(j)+\cdots+\dmax^{r_j-1}(j))\leq3r_j\dav(j)$, which is bounded by $D_j$ above.
\item When there are $r_j$ open facilities in $B_j'$ but at most $r_j-1$ in $\calF_j\setminus \calF_{j,r_j}$, we first connect $j$ to open facilities in $U_{j,1},\dots,U_{j,r_j-1}$, whose total cost is at most $3\sum_{p=1}^{r_j-1}\dmax^p(j)\leq3r_j\dav(j)$. For the $r_j$-th connection, since there exists one in $B_j'\cap\calF_{j,r_j}$ distinct from the previous $r_j-1$ ones, the cost of this connection is at most $A_j$ with probability 1.
\item When there are $r_j-1$ open facilities in $B_j'$, we connect $j$ to the $r_j-1$ open facilities in $U_{j,1},\dots,U_{j,r_j-1}$ with a cost upper-bounded by $3\sum_{p=1}^{r_j-1}\dmax^p(j)\leq3r_j\dav(j)$. There is another open facility in $U_{j,r_j}$, which is at most $3\dmax(j)$ away using~\cref{lem-bundle}. Since $Q_j=3\dmax(j)$ in this case, we can bound this stochastic cost by $Q_j$.
\end{enumerate}

\subparagraph{(C) $\dmin^{r_j}(j)\leq T<\dmax(j)/15$.} We define another random variable $A'_j$, which is the total cost of connecting $j$ to \emph{all} open facilities in $B'_j\setminus\ball(j,T)$. Likewise, the expectation is at most
\begin{align*}
\E[A'_j]&=\sum_{i\in B_j'\setminus \ball(j,T)}x_{ij}d(i,j)=\sum_{i\in B_j'\setminus \ball(j,T)}x_{ij}\calL_T(i,j)\\
&\leq\sum_{i\in\calF_{j,r_j}}x_{ij}\calL_T(i,j)=\davt{T}^{r_j}(j),
\end{align*}
where the inequality is due to $T\geq\dmin^{r_j}(j)$ and thus $B_j'\setminus\ball(j,T)\subseteq\calF_{j,r_j}$.
We also notice $T<\dmax(j)/15$, thus one has
    \[\davt{T}^{r_j}(j)\geq\davt{\dmax(j)/15}^{r_j}(j)\geq\sum_{\substack{i\in\calF_{j,r_j}\\d(i,j)>\dmax(j)/15}}x_{ij}d(i,j)
    \geq\frac{\dmax(j)}{15}(r_j-y(B_j)).\]
    Therefore, the expectation of $Q_j$ satisfies
    \[\E[Q_j]=3\dmax(j)(r_j-y(B_j'))\leq3\dmax(j)(r_j-y(B_j))\leq45\davt{T}^{r_j}(j).\]
In this case, we charge $D_j=3r_j\dav(j)+T$ and $X_j=A'_j+Q_j$, therefore resulting in $\E[X_j]\leq46\davt{T}^{r_j}(j)$. We consider three cases that are similar to the above.
\begin{enumerate}
	\item When there are $r_j$ open facilities in $\ball(j,T)$, the $r_j$-th connection is at most $T$. We charge $T$ to $D_j$, as well as $3r_j\dav(j)$ to $D_j$ for the first $r_j-1$ connections in $\bigcup_{p<r_j}U_{j,p}$.
	\item When there are $r_j$ in $B_j'$ but at most $r_j-1$ in $\ball(j,T)$, we connect $r_j-1$ facilities in $\bigcup_{p<r_j}U_{j,p}$, and the $r_j$-th connection cost in $B_j'\setminus \ball(j,T)$ is bounded by $A'_j$.
    \item When there are $r_j-1$ facilities open in $B_j'$, there is another facility at most $3\dmax(j)$ away in $U_{j,r_j}$, hence the cost is bounded by $Q_j$. 
\end{enumerate} 

\subparagraph{Summary for $D'$.} Since $T$ can only satisfy one of the conditions above, by taking the maximum, we obtain $D_j\leq3r_j\dav(j)+45T$ and $\E[X_j]\leq 46\davt{T}^{r_j}(j)$, for each $j\in D'$.

\paragraph{{(III)} For $j\in D\setminus D'$.}\label{proof:667-3}
According to the construction of $D'$, there exists $j'\in D\setminus D'$ s.t. $r_j=r_{j'}=r,\,d(j,j')\leq6\max\{\dav^{r}(j),\dav^{r}(j')\}$ and $\dav^r(j')\leq\dav^r(j)$. We still use the bundles $U_{j,1},\dots,U_{j,r-1}$ to satisfy the first $r-1$ connections of $j$ and use another open facility that serves $j'$ to serve $j$. The first $r-1$ bundles in $\que_j$ incur service cost at most $3r\dav(j)$. For the $r$-th connection, we first fix two constant parameters $\alpha,\beta>6,\alpha+1<\beta$ that are determined later.
\subparagraph{(A) $d(j,j')<\alpha T$.} We consider the following three cases.
\begin{enumerate}
	\item If $\beta T\geq\dmax(j')/15$, since the farthest $r$-th open facility is at most $3\dmax(j')$ away from $j'$, it is also at most $d(j,j')+3\dmax(j')\leq(\alpha+45\beta)T$ away from $j$ using the triangle inequality. We charge $3r\dav(j)+(\alpha+45\beta)T$ to $D_j$.
	\item If $\beta T<\dmin^r(j')$, we reuse random variables $A_{j'}$ and $Q_{j'}$ defined above. We claim
    \begin{align}\E[A_{j'}]&\leq\sum_{i\in B_{j'}'\cap\calF_{j',r}}x_{ij'}d(i,j')\leq\dav^r(j')\leq\dav^r(j)=\davt{T}^r(j),\label{eqn:expectation:Aprime}\\
    \E[Q_{j'}]&\leq3\dmax(j')(r-y(B_{j'}'))\leq45\dav^r(j')\leq45\dav^r(j)=45\davt{T}^r(j),\label{eqn:expectation:Qprime}
    \end{align}
    where the equal signs are due to $y(\calF_{j,r}\cap\ball(j,T))=0$. Indeed, since $\beta>\alpha+1$ and $d(j,j')<\alpha T$, one has $T+d(j,j')<\beta T$ and thus $\ball(j,T)\subseteq\ball(j',\beta T)$. For the sake of contradiction, if $y(\calF_{j,r}\cap\ball(j,T))>0$, one has $y(\ball(j,T))>y(\calF_j\setminus\calF_{j,r})=r-1$, which implies $y(\ball(j',\beta T))>r-1$. But this is impossible since $\dmin^r(j')>\beta T$, hence \eqref{eqn:expectation:Aprime} and \eqref{eqn:expectation:Qprime} follow. We further have the following cases.
    \begin{enumerate}
    	\item If there are $r$ open facilities in $\calF_{j'}\setminus\calF_{j',r}$, we connect the facilities in $U_{j,1},\dots,U_{j,r-2}$ to $j$, and connect the (at least) remaining 2 in $\calF_{j'}\setminus \calF_{j',r}$, both at most $\alpha T+\dmax^{r-1}(j')\leq\alpha T+\dav^r(j')\leq\alpha T+\dav^r(j)$ away from $j$. The overall charge to $D_j$ is at most $3(\dav^1(j)+\cdots+3\dav^r(j))+2\alpha T=3r\dav(j)+2\alpha T$.
    	\item If there are $r$ open facilities in $B_{j'}'$ but at most $r-1$ in $\calF_{j'}\setminus \calF_{j',r}$, we connect the open facilities in $U_{j,1},\dots,U_{j,r-1}$ to $j$. For the $r$-th connection which can be found in $B_{j'}'$, it is at most $A_{j'}+\alpha T$ away from $j$, so we charge $\alpha T$ to $D_j$ and $A_{j'}$ to $X_j$, increasing $\E[X_j]$ by at most $\davt{T}^r(j)$ using~\eqref{eqn:expectation:Aprime}. 
    	\item If there are $r-1$ open facilities in $B_{j'}'$, we connect the facilities in $U_{j,1},\dots,U_{j,r-1}$ to $j$. Notice there are $r$ open facilities in $U_{j',1},\dots,U_{j',r}$, therefore there exists an additional open facility at most $Q_{j'}+\alpha T$ away from $j$ using~\cref{lem-bundle}. We charge $\alpha T$ to $D_j$ and $Q_{j'}$ to $X_j$, increasing $\E[X_j]$ by at most $45\davt{T}^r(j)$ using~\eqref{eqn:expectation:Qprime}.
    \end{enumerate}

Therefore, we charge at most $3r\dav(j)+2\alpha T$ to $D_j$ and $A_{j'}+Q_{j'}$ to $X_j$, and $\E[X_j]$ is increased by at most $46\davt{T}^r(j)$.

\item Finally, if $\dmin^r(j')\leq\beta T<\dmax(j')/15$, we have the following cases.
\begin{enumerate}
	\item If there are $r$ open facilities in $\ball(j',\beta T)$, we charge $(\alpha+\beta)T$ to $D_j$ for the $r$-th connection.
	\item If there are $r$ open facilities in $B_{j'}'$ but at most $r-1$ in $\ball(j',\beta T)$, define the random variable $P_{j'}$ as the total cost of connecting $j'$ to all open facilities in $B_{j'}'\setminus \ball(j',\beta T)$. The $r$-th connection cost for $j$ is at most $P_{j'}+\alpha T$ using the triangle inequality, so we directly charge $P_{j'}$ to $X_j$ and charge $\alpha T$ to $D_j$.
	\item If there are $r-1$ open facilities in $B_{j'}'$, there exists a facility at most $3\dmax(j')$ away from $j'$. Define the random variable $M_{j'}$ as the cost of connecting $j'$ to this open facility (0 if there are $r$ open facilities in $B_{j'}'$). We charge $\alpha T$ to $D_j$ and $M_{j'}$ to $X_j$.
\end{enumerate}

In the three cases above, we charge at most $3r\dav(j)+(\alpha+\beta)T$ to $D_j$ and $P_{j'}+M_{j'}$ to $X_j$. Next we need to bound the expectation of $P_{j'}$ and $M_{j'}$. We first obtain
\begin{align}
\E[P_{j'}]&=\sum_{i\in B_{j'}'\setminus \ball(j',\beta T)}x_{ij'}d(i,j')\leq\sum_{i\in\calF_{j'}\setminus \ball(j',\beta T)}x_{ij'}d(i,j')\notag\\
&\leq\sum_{i\in\calF_{j}\setminus \ball(j',\beta T)}x_{ij}d(i,j')\notag\\
&\leq\frac{\beta}{\beta-\alpha}\sum_{i\in\calF_{j}\setminus \ball(j,T)}x_{ij}d(i,j)=\frac{\beta}{\beta-\alpha}\sum_{p\in[r]}\davt{T}^p(j),\label{eqn:expectation:P}
\end{align}
where the second inequality is because $x$ is optimal for $j'$, so the amount of assignment of $j'$ outside $\ball(j',\beta T)$ is at most that of $j$ outside $\ball(j',\beta T)$, and replacing $x_{ij'}$ with $x_{ij}$ cannot make the weighted sum decrease. The third inequality is because for any $i$ such that $d(i,j')>\beta T$, we have
\[\frac{d(i,j')}{d(i,j)}\leq\frac{d(i,j')}{d(i,j')-d(j,j')}\leq\frac{d(i,j')}{d(i,j')-\alpha T}<\frac{\beta}{\beta-\alpha},\]
and $\alpha+1<\beta,d(j,j')<\alpha\Rightarrow\ball(j,T)\subseteq\ball(j',\beta T)$. We note that a similar argument can be found in~\cite{byrka2018constant}.
For $M_{j'}$, because $\dmin^r(j')\leq\beta T<\dmax(j')/15$, we obtain
\[\davt{\beta T}^r(j')\geq\davt{\frac{\dmax(j')}{15}}^r(j')\geq\sum_{\substack{i\in\calF_{j',r}\\d(i,j')>\frac{\dmax(j')}{15}}}x_{ij}d(i,j)
\geq\frac{\dmax(j')}{15}(r-y(B_{j'})).\]
One then has,
\begin{align}
	\E[M_{j'}]&\leq3\dmax(j')(r-y(B_{j'}'))\leq45\davt{\beta T}^r(j')\notag\\
	&=45\sum_{i\in\calF_{j'}\setminus \ball(j',\beta T)}x_{ij'}d(i,j')
	\leq45\sum_{i\in\calF_{j}\setminus \ball(j',\beta T)}x_{ij}d(i,j')\notag\\
	&\leq\frac{45\beta}{\beta-\alpha}\sum_{i\in\calF_{j}\setminus \ball(j,T)}x_{ij}d(i,j)=\frac{45\beta}{\beta-\alpha}\sum_{p\in[r]}\davt{T}^p(j),\label{eqn:expectation:M}
\end{align}
and thus $\E[X_j]\leq\frac{46\beta}{\beta-\alpha}\sum_{p\in[r]}\davt{T}^p(j)$ follows by combining~\eqref{eqn:expectation:P} and~\eqref{eqn:expectation:M}.
\end{enumerate} 

\subparagraph{(B) $d(j,j')\geq\alpha T$.} We have
\[\alpha T\leq d(j,j')\leq6\dav^r(j)\leq6\davt{T}^r(j)+6T\Rightarrow T\leq\frac{6}{\alpha-6}\davt{T}^r(j).\]

Using the triangle inequality, since $d(j,j')\leq6\davt{T}^r(j)+6T$, we first charge $6T$ to $D_j$ and $6\davt{T}^r(j)$ to $X_j$ with probability 1.
For the $r$-th connection from facilities that serve $j'$, by reusing $A_{j'}$ and $Q_{j'}$ and charging them to $X_j$, $\E[X_j]$ is further increased by at most
\begin{align*}
    &\sum_{i\in B_{j'}'\cap\calF_{j',r}}x_{ij'}d(i,j')+3\dmax(j')(r-y(B_{j'}'))\\
    &\leq46\dav^r(j')\leq46\dav^r(j)\leq46T+46\davt{T}^r(j),
\end{align*}
and since $T\leq6\davt{T}^r(j)/(\alpha-6)$, we have
\[\E[X_j]\leq\left(52+46\cdot\frac{6}{\alpha-6}\right)\davt{T}^r(j)=\frac{52\alpha-36}{\alpha-6}\davt{T}^r(j).\]

\subparagraph{Summary for $D\setminus D'$.} We charge $D_j\leq3r_j\dav(j)+(\alpha+45\beta)T$ and charge up $X_j$ such that $\E[X_j]\leq\max\left\{\frac{46\beta}{\beta-\alpha},\frac{52\alpha-36}{\alpha-6}\right\}\sum_{p\in[r_j]}\davt{T}^p(j)$ for each $j\in D\setminus D'$.

\paragraph{(IV) Conclusion.}
To summarize all the discussion above, we take the maximum of all values for $\calC\setminus D,\,D'$ and $D\setminus D'$, obtaining the following for each $j\in\calC$,
\begin{align*}D_j&\leq3r_j\dav(j)+\max\{135,\alpha+45\beta\}\cdot T,\\
\E[X_j]&\leq \max\left\{135,\frac{46\beta}{\beta-\alpha},\frac{52\alpha-36}{\alpha-6}\right\}\cdot\sum_{i\in\calF}x_{ij}\calL_T(i,j).\end{align*}

By letting $\beta=9.51,\alpha=7.58$, we have the lemma.
\end{proof}
